\documentclass[11pt]{article}

\usepackage{amsmath, amsthm, amssymb,color,fullpage,url,booktabs}  %cite
\usepackage{graphicx}

\ifdefined\blackandwhite
\usepackage[pdftex]{hyperref} %pagebackref
\else
\usepackage[colorlinks,pdftex]{hyperref} %pagebackref
\hypersetup{citecolor = blue, linkcolor = red}
\fi

\newcommand{\nc}{\newcommand}
\nc{\rnc}{\renewcommand}

\newcommand{\bra}[1]{\left\langle #1\right|}
\newcommand{\ket}[1]{\left|#1\right\rangle}
\newcommand{\proj}[1]{\left|#1\right\rangle\left\langle #1\right|}
\newcommand{\braket}[2]{\left\langle #1\middle|#2\right\rangle}

\DeclareMathOperator{\tr}{tr}

\def\be#1\ee{\begin{equation}#1\end{equation}}
\def\bea#1\eea{\begin{eqnarray}#1\end{eqnarray}}
\def\beas#1\eeas{\begin{eqnarray*}#1\end{eqnarray*}}
\def\ba#1\ea{\begin{align}#1\end{align}}
\def\bas#1\eas{\begin{align*}#1\end{align*}}
\def\bpm#1\epm{\begin{pmatrix}#1\end{pmatrix}}
\nc{\non}{\nonumber}
\nc{\nn}{\nonumber}
\nc{\eq}[1]{(\ref{eq:#1})}
\nc{\eqs}[2]{(\ref{eq:#1}) and (\ref{eq:#2})}
\rnc{\L}{\left} 
\nc{\R}{\right}
\nc{\ra}{\rightarrow}
\nc{\ot}{\otimes}
\nc{\grad}{{\vec{\nabla}}}

\newtheorem{thm}{Theorem}
\newtheorem*{thm*}{Theorem}

\newtheorem{proto}{Protocol}

\theoremstyle{definition}
\newtheorem{remark}{Remark}

\newtheorem{dfn}[thm]{Definition}
\theoremstyle{plain}

\makeatletter
\newtheorem*{rep@theorem}{\rep@title}
\newcommand{\newreptheorem}[2]{%
\newenvironment{rep#1}[1]{%
 \def\rep@title{#2 \ref{##1} (restatement)}%
 \begin{rep@theorem}}%
 {\end{rep@theorem}}}
\makeatother

\newreptheorem{thm}{Theorem}
\newreptheorem{lem}{Lemma}

\nc\eps{\epsilon}

\nc\cA{\mathcal{A}}
\nc\cB{\mathcal{B}}
\nc\cC{\mathcal{C}}
\nc\cD{\mathcal{D}}
\nc\cE{\mathcal{E}}
\nc\cF{\mathcal{F}}
\nc\cG{\mathcal{G}}
\nc\cH{\mathcal{H}}
\nc\cI{\mathcal{I}}
\nc\cJ{\mathcal{J}}
\nc\cK{\mathcal{K}}
\nc\cL{\mathcal{L}}
\nc\cM{\mathcal{M}}
\nc\cN{\mathcal{N}}
\nc\cO{\mathcal{O}}
\nc\cP{\mathcal{P}}
\nc\cQ{\mathcal{Q}}
\nc\cR{\mathcal{R}}
\nc\cS{\mathcal{S}}
\nc\cT{\mathcal{T}}
\nc\cU{\mathcal{U}}
\nc\cV{\mathcal{V}}
\nc\cW{\mathcal{W}}
\nc\cX{\mathcal{X}}
\nc\cY{\mathcal{Y}}
\nc\cZ{\mathcal{Z}}

\nc\bbC{\mathbb{C}}

\nc\bbF{\mathbb{F}}
\nc\bbM{\mathbb{M}}
\nc\bbN{\mathbb{N}}
\nc\bbR{\mathbb{R}}
\nc\bbZ{\mathbb{Z}}

\nc\benum{\begin{enumerate}}
\nc\eenum{\end{enumerate}}
\nc\bit{\begin{itemize}}
\nc\eit{\end{itemize}}

\nc{\todo}[1]{\textcolor{red}{todo: #1}}
\nc{\Anote}[1]{\textcolor{red}{Aram note: #1}}

\def\begsub#1#2\endsub{\begin{subequations}\label{eq:#1}\begin{align}#2\end{align}\end{subequations}}
\nc\qand{\qquad\text{and}\qquad}
\nc\mnb[1]{\medskip\noindent{\bf #1}}

\nc{\pder}[2]{\frac{\partial {#1}}{\partial {#2}}}
\nc{\p}{\partial}

\usepackage{fullpage}
\usepackage{amsmath,amsfonts,amsthm,mathrsfs,mathpazo,xspace,hyperref,graphicx}
\usepackage{endnotes}
\usepackage{color}
\usepackage{bm}
\usepackage{times}
\usepackage{amssymb,latexsym}
\usepackage{enumitem}

\usepackage{algcompatible}
\usepackage{algorithm}
\usepackage[noend]{algpseudocode}
\usepackage{verbatim}

\newtheorem{theorem}{Theorem}
\newtheorem*{theorem*}{Theorem}

\newtheorem{lemma}[theorem]{Lemma}
\newtheorem*{lemma*}{Lemma}

\newtheorem{fact}[theorem]{Fact}
\newtheorem*{fact*}{Fact}
\newtheorem{corollary}[theorem]{Corollary}
\newtheorem{definition}[theorem]{Definition}

\newcommand{\beq}{\begin{eqnarray}}
\newcommand{\eeq}{\end{eqnarray}}

\newcommand{\Tr}{\mbox{\rm Tr}}

\newcommand{\ssmat}{subset-matrix}
\newcommand{\ssmats}{subset-matrices}
\newcommand{\vblock}{\varphi_{\text{block}} }
\newcommand{\vfar}{\varphi_{\text{far}} }

\newcommand{\keyexponent}{3Q+12N+24}

\newcommand{\Q}{\eta}
\newcommand{\dearth}[2]{d_{\infty}(#1, #2)}
\newcommand{\dearthsmooth}[3]{d^{#1}_{\infty}(#2, #3)}
\newcommand{\ups}{\upsilon}
\newcommand{\SR}{rk_{Schmidt}}

\newcommand{\B}{B}

\newcommand{\Qovereps}{Q/\epsilon+\log(1/\epsilon)/\epsilon}

\begin{document}

\title{Universality of EPR pairs in Entanglement-Assisted
	Communication Complexity, and the Communication Cost of State
	Conversion}
\author{Matthew Coudron\thanks{Institute for Quantum Computing, University of Waterloo \texttt{\href{mailto:mcoudron@uwaterloo.ca}{\color{black}mcoudron@uwaterloo.ca}}. } 
\qquad Aram W. Harrow\thanks{Center for Theoretical Physics, MIT. {\tt aram@mit.edu}} }

\date{}

\maketitle

\thispagestyle{empty}

\begin{abstract}

  In this work we consider the role of entanglement assistance in quantum communication
  protocols, focusing, in particular, on whether the type of shared entangled state can
  affect the quantum communication complexity of a function.  This question is interesting
  because in some other settings in quantum information, such as non-local games, or tasks that involve quantum communication between players and
  referee, or simulating bipartite unitaries or communication
  channels, maximally entangled states are known to be less useful as a resource than some partially
  entangled states.  By contrast, we prove that the bounded-error entanglement-assisted quantum communication complexity of a partial or total function cannot be
  improved by more than a constant factor by replacing maximally entangled states with
  arbitrary entangled states. In particular, we show that every quantum communication
  protocol using $Q$ qubits of communication and arbitrary shared entanglement can be
  $\epsilon$-approximated by a protocol using $O(\Qovereps)$ qubits of communication and
  \emph{only} EPR pairs as shared entanglement. This conclusion is opposite of
  the common wisdom in the study of non-local games, where it has been shown, for example,
  that the I3322 inequality has a non-local strategy using a non-maximally entangled
  state, which surpasses the winning probability achievable by any strategy using a
  maximally entangled state of any dimension \cite{VW11}.  
  We leave open the question of how much the use of a shared maximally entangled state can reduce the quantum communication complexity of a function.

  Our second result concerns an old question in quantum information theory: How much
  quantum communication is required to approximately convert one pure bipartite entangled
  state into another?  We give simple and efficiently computable upper and lower bounds.  Given two bipartite states $\ket \chi$ and
  $\ket \upsilon$, we define a natural quantity, $\dearth{\ket \chi}{\ket \upsilon}$,
  which we call the $\ell_{\infty}$ Earth Mover's distance, and we show
  that the communication cost of converting between $\ket \chi$ and $\ket \upsilon$ is
  upper bounded, up to a constant multiplicative factor, by $\dearth{\ket \chi}{\ket \upsilon}$.  Here $\dearth{\ket \chi}{\ket \upsilon}$ may be informally described as the minimum over
  all transports between the log of the Schmidt coefficients of $\ket \chi$ and those of
  $\ket \upsilon$, of the maximum distance that any amount of mass must be moved in that
  transport.  A precise definition is given in the introduction.  Furthermore, we prove a complementary lower bound on the cost of state
  conversion by the $\epsilon$-Smoothed $\ell_{\infty}$-Earth Mover's Distance, which is a natural smoothing of the $\ell_{\infty}$-Earth Mover's Distance that we will define via a connection with optimal transport theory.

\end{abstract}

\pagebreak
\pagenumbering{arabic}

\section{Introduction}

    \subsection{Entanglement-assisted communication complexity}
    
Imagine that two cooperating players, Alice and Bob, are given the task of evaluating a function $f(x,y)$ ($x,y\in\{0,1\}^n$), where $x$ is known only to Alice and $y$ is known only to Bob.  The communication complexity of $f$ is
the number of bits that Alice and Bob need to exchange in order to compute $f$.  Popular variations
of this framework include allowing a small probability of error, allowing qubits to be communicated instead
of classical bits, and allowing extra resources such as shared randomness or entanglement.

    In classical communication complexity, Newman's theorem states that arbitrarily large amounts of
    shared randomness in a protocol can be replaced by a distribution with $O(\log(n/\eps))$ bits of
    entropy while only reducing the success probability of that protocol by $\eps$. (Here $n$ is
    the input size of each party.)  Is there a quantum analogue to this result?
    
    In one sense the answer is ``no".  Given a two-party entanglement-assisted protocol for, say,
    computing the value of some function, we cannot replace the shared entanglement with some
    different, less entangled, state, without causing large errors~\cite{JRS08,AHLNSZ}.  It is
    an open question whether it is possible to replace a large entangled state with a less
    entangled one while also changing the communication protocol.
    
    However, while it remains a challenge to characterize the \emph{dimension} of shared entanglement required for optimal entanglement-assisted quantum communication protocols, in this work we show that the \emph{type} of shared entanglement required by such protocols can be neatly characterized.  In Theorem \ref{thm::weakplusearth} below, we establish that the bounded-error entanglement-assisted quantum communication complexity of a partial or total function cannot be improved by more than a constant factor by replacing maximally entangled states with arbitrary entangled states.  This is accomplished by constructing an explicit protocol which allows two parties, who only share maximally entangled states, to simulate any entanglement-assisted quantum communication task regardless of the shared state that that task originally required.
    
    \begin{theorem} \label{thm::weakplusearth} Consider a quantum communication protocol $\mathcal{R}$ whose goal is to compute a joint function $f(x,y) \in \{0, 1\}$. Suppose that $\mathcal{R}$ uses an arbitrary bipartite entangled state $\ket{\psi}^{AB}$ (of unbounded dimension), as well as $Q$ qubits of communication total, in either direction (for sufficiently large $Q \geq 15$).  Then, for every $\eps > 0$, there exists a quantum communication protocol $\mathcal{R'}$ which simulates $\mathcal{R}$ with error $\epsilon$, while using only a maximally entangled state as an entangled resource (rather than $\ket{\psi}^{AB}$ or any other state), and using $O(\Qovereps)$ qubits of communication.  Thus, if $\mathcal{R}$ computes $f$ with error $\eps '$ it follows that $\mathcal{R'}$ computes $f$ with error $\epsilon+\eps '$.
    \end{theorem}
    
    Theorem \ref{thm::weakplusearth} shows that, although the role of shared entanglement
    in quantum communication complexity is still not well understood, the \emph{type} of
    shared entanglement does not drastically change communication complexity.  This is
    true regardless of input size or promise, as long as we are in the constant-error
    regime and some communication is allowed between players (unlike, say, the
    simultaneous-message-passing model). This result sets quantum communication complexity
    apart from settings such as channel simulation~\cite{BDHSW-qrst}, nonlocal
    games~\cite{JungeP11,Regev12}, unitary gate simulation~\cite{HL07}, and communication
    tasks involving quantum communication between referees and players~\cite{LeungTW13}.
    In each of those cases the ratio between the EPR-assisted costs and the (unrestricted)
    entanglement-assisted costs can be made arbitrarily large.  This suggests that the
    role of shared entanglement in quantum communication complexity may be fundamentally
    different than in these other settings.  Furthermore, the result achieved in Theorem
    \ref{thm::weakplusearth} may be useful in future work attempting to further bound the
    role of entanglement in quantum communication complexity, as it restricts the problem
    to the case of shared EPR pairs, without loss of generality.
    
    It may be worth noting that the proof of Theorem \ref{thm::weakplusearth} is nearly
    oblivious to the entanglement-assisted protocol being considered in the following
    sense: Given a protocol $\cP$ using $Q$ qubits of communication and a shared entangled
    state $\ket\psi$, we can replace $\ket\psi$ with a ``consolidated" state $\rho$ at the
    cost of error $\eps$. Moreover, $\rho$ can be prepared from a maximally entangled
    state using $O(\Qovereps)$ communication.  Taking $\eps$ constant implies that the
    EPR-assisted communication complexity of a function is at most $O(1)$ times the
    (unrestricted) entanglement-assisted communication complexity of that function.  It
    was not necessary to modify the protocol $\cP$ to achieve this result, except to
    pre-compose it with a pre-processing protocol which starts with only EPR pairs, and
    prepares the state $\rho$ using only $O(\Qovereps)$ communication.  $\cP$ can then be
    run on $\rho$ directly.  Such a protocol-agnostic preprocessing should not be taken
    for granted, since it is known that reducing the number of EPR pairs may in some cases
    require more than just pre-processing~\cite{JRS08,AHLNSZ}.
        
  \subsection{Communication cost of state transformations}

  Our second contribution, which is related at the level of techniques to Theorem \ref{thm::weakplusearth}, is to provide upper and lower bounds for an old quantity studied in quantum information theory, the communication cost of state transformation.
   
  Suppose that $\ket\chi^{AB}$ and $\ket\nu^{AB}$ are bipartite pure quantum states, with vectors of
  Schmidt coefficients denoted respectively by $\chi$ and $\nu$.  In this setting
  it is known that $\ket\chi$ can be exactly converted into $\ket\nu$ using LOCC if and only if $\chi$
  is majorized by $\nu$~\cite{Nielsen99a}.  But the communication cost of this transformation is known only in
  a few special cases. If $\ket\chi = \ket{\chi_0}^{\otimes n}$ and $\ket{\nu} =
  \ket{\nu_0}^{\otimes n}$ for some states $\ket{\chi_0},\ket{\nu_0}$, then this cost is
  $O(\sqrt n)$ or less in some special cases (e.g. $\ket{\nu_0}$ is maximally entangled).
  More generally there is, in principle, an exact characterization of the communication cost (either LOCC, or quantum communication)
  of state transformation using the Schubert calculus due to Daftuar and Hayden
  \cite{DH05}, but in practice it is difficult to extract concrete bounds from their
  main theorem.
  
  In this work we identify a simple and efficiently computable quantity, which we call the
  $\ell_{\infty}$ Earth Mover's  (or Wasserstein) Distance, which tells us approximately
  how much quantum communication is required to
  transform $\ket\chi$ to $\ket\nu$.  Given its simple form, we believe that this
  quantity may be a useful tool in quantum information theory.  
%The $\ell_{\infty}$ Earth Mover's Distance is defined as follows:
  
  \begin{definition}[$\ell_{\infty}$ Earth Mover's Distance ]\label{def::earthmoverdistance}
   
   Let $\ket\chi^{AB} = \sum_{i\in X} \sqrt{\chi_i} \ket i^A \ot \ket i^B$ and $\ket\upsilon^{AB} = \sum_{j\in Y} \sqrt{\upsilon_j} \ket j^A \ot \ket j^B$ be two states.  We define $\dearth{\ket \chi}{\ket \upsilon}$ to be the $\ell_{\infty}$ Earth Mover's distance between $\ket\chi$ and $\ket\upsilon$, which is equal to the minimum $\mu \geq 0$ for which there exists a joint distribution $\omega(x,y): X \times Y \to \mathbb{R}_{\geq 0}$ such that:
   
   \begin{itemize}
   	\item  $\sum_{j \in Y} \omega(i,j) = \chi_i$ $\forall i \in X$
   	\item $\sum_{i \in X} \omega(i,j) = \ups_j$ $\forall j \in Y$ 
   	\item $\omega(i,j) = 0$ \text{ whenever } $| \log(\chi_i) - \log(\ups_j)| > \mu$
   \end{itemize}
  \end{definition}
  
We can think of $\chi$ as corresponding to placing $\chi_i$ mass at position
$\log(\chi_i)$ for each $i$, and similarly for $\ups$.  Then $\dearth{\ket \chi}{\ket
  \upsilon}$ is the $\ell_{\infty}$ EMD (Earth Mover's distance) between these distributions.

  In Section \ref{sec::Earthmoverstatetransformation} we will show that this quantity gives an intuitive upper bound on the amount of quantum communication required to transform one bipartite shared state into another.  In particular  we prove the following theorem.
  
  \begin{theorem} \label{thm::earthmover}     Let $\ket\chi^{AB}$ and $\ket\upsilon^{AB}$ be two bipartite shared states.  There is a protocol $\mathcal{M}_{\chi\rightarrow\upsilon}$ which can prepare $\ket\upsilon$ from $\ket\chi$, using only $4 \lceil \dearth{\ket \chi}{\ket \upsilon} \rceil + 8$ qubits of communication.
  \end{theorem}

In Section \ref{sec::lowerbound} we establish a complementary lower bound, showing that a ``$\epsilon$-smoothed" version of the $\ell_{\infty}$ Earth Mover's Distance, denoted by $\dearthsmooth{\eps}{\ket{\chi}}{\ket{\upsilon}}$, gives a lower bound on the cost of state transformation.  That is:

\def\ThmLowerBound
 { 	Given any two bipartite shared states $\ket {\psi}^{AB} = \sum_{i} \sqrt{\psi_i} \ket i^A \ot \ket i^B$ and $\ket {\phi}^{AB} = \sum_{i} \sqrt{\phi_i} \ket i^A \ot \ket i^B$, shared between two parties $A$ and $B$, together with a unitary $U_\mathcal{P}$ which can be performed on the state $\ket \psi^{AB}$ via a quantum communication protocol $\mathcal{P}$, that uses $Q$ qubits of communication between $A$ and $B$, we have that, for every $\epsilon$:
  	
  	\[  \left |\bra{ \phi}^{AB}U_\mathcal{P}\ket{\psi}^{AB} \right |  \leq  1 - \frac{1}{4} \epsilon^2  + 24 \cdot  2^{-\frac{1}{2}(\dearthsmooth{\eps}{\ket \psi}{\ket \phi} - 3Q)}\]
}
  \begin{theorem} \label{thm::lowerbound}
\ThmLowerBound
  \end{theorem}
   In words:  If two shared states cannot be brought within small  $\ell_{\infty}$ Earth Mover's Distance of each other by moving an $\eps$ quantity of mass of their Schmidt coefficients, then they also cannot be brought closer than $1-O(\eps^2)$ fidelity with each other without using  $\Omega(\dearthsmooth{\eps}{\ket \psi}{\ket \phi})$ qubits of communication (for sufficiently large values of $\dearthsmooth{\eps}{\ket \psi}{\ket \phi}$).  Thus, the $\epsilon$-smoothed $\ell_{\infty}$ Earth Mover's Distance provides a lower bound on the communication cost of state conversion.  On the other hand, from the definition of $\dearthsmooth{\eps}{\ket \psi}{\ket \phi}$, stated in Definition \ref{def:epssmoothdist}, we note here that one can use Theorem \ref{thm::earthmover} to move $\ket \psi$ to within $1 - \epsilon$ fidelity of $\ket \phi$ using only $O(\dearthsmooth{\eps}{\ket \psi}{\ket \phi})$ qubits of communication.  To do this, omit the $\epsilon$ mass of Schmidt coefficients on which the two states have large $\epsilon$-smoothed $\ell_{\infty}$ distance, and apply Theorem \ref{thm::earthmover} as one would do with the regular $\ell_{\infty}$ Earth Mover's Distance.  In this sense $\dearthsmooth{\eps}{\ket \psi}{\ket \phi}$ gives both an upper and lower bound on the communication cost of state conversion.

To put these bounds in context:  One could consider entanglement concentration and dilution to be the starting point for the study of state conversion.  The original paper on entanglement concentration and dilution~\cite{BBPS96} concerned the
many-copy limit and did not attempt to bound the amount of classical communication used.
The first time the classical communication cost of state conversion was considered explicitly seems
to have been in \cite{LP99}, which could be said to establish a version of our upper bound
in the case where the starting state is maximally entangled.  (Their result is not quite
that general but contains many of the key ideas.)  A version of our lower bound was
established, again for the case of starting with maximally entangled states, in
\cite{HL02,HW02}.  These lower bounds could be applied to general state conversion but
relied on R\`enyi entropy inequalities that are clearly not tight in many cases.  Finally, as noted earlier, a full
characterization of the communication cost of general state conversion was given in
\cite{DH05} but the resulting formula is complicated and there is not an efficient
algorithm known to evaluate it.

We conclude the section with two remarks about notation.

\begin{remark}
	In theorem statements above, and where appropriate, we have made use of
        superscripts $A$ and $B$, as in $\ket {\psi}^{AB} = \sum_{i} \sqrt{\psi_i} \ket
        i^A \ot \ket i^B$ to explicitly denote the two halves of the bipartite division of
        a state.  However, since all of the shared entangled states considered in this
        paper are bipartite, and since the two components of the bipartite division are
        generally clear from context, we will usually omit this notation.
\end{remark}

\begin{remark}  \label{rem:decomp}
When considering a bipartite state $\ket \psi$, we will assume that the state has a Schmidt decomposition of the form $\ket {\psi} = \sum_{i} \sqrt{\psi_i} \ket i \ot \ket i$ across the implicit bipartite division.  This is done in the theorem statements above and everywhere in the paper.  We can assume this WLOG because any state that has the same Schmidt coefficients as $\ket \psi$ can be moved to this canonical form (and vice versa) using only local unitary transformations, which can be implemented with no quantum communication between the two components of the bipartite division.  Thus our analysis of communication costs is unaffected by assuming WLOG that, in any quantum communication protocol, shared entangled states start and end in this form.
\end{remark}

\section{Entanglement-Assisted Communication Complexity} \label{sec::mainresult}

In this section we will discuss the proof of our main result, Theorem
\ref{thm::weakplusearth}, which shows that arbitrary entanglement-assisted quantum
communication protocols can be simulated by quantum communication protocols that use only
the maximally entangled state as an entangled resource.  
A basic fact we will need is that two bipartite pure states which are sufficiently different in the
distribution of mass across their Schmidt coefficients must be nearly orthogonal.  This
fact is stated for our specific purposes in Lemma \ref{lem:innerprod}
below.  
Crucially, such states
\emph{remain} nearly orthogonal even after one of them is acted on by any unitary which
can be implemented with a small amount of quantum communication, as we detail in
Lemma \ref{lem:cominnerprod}.

\begin{lemma}\label{lem:cominnerprod}
	Given two quantum states $\ket \psi$ and $\ket \nu$ on $\mathcal{H}_A \otimes \mathcal{H}_B$, such that the Schmidt coefficients of $\psi$ are upper bounded by $\lambda_{\max}$, and those of $\nu$ are upper bounded by $\nu_{\max}$, and further given a unitary transformation $\mathcal{U}$ on $\mathcal{H}_A \otimes \mathcal{H}_B$ which can be implemented using at most $Q$ qubits of communication between the $\mathcal{H}_A$ and $\mathcal{H}_B$ components of the Hilbert space, it follows that:
	
	\[ |  \bra{\psi}\mathcal{U}\ket{\nu}| \leq   2^{\frac{3}{2}Q}\cdot \SR(\ket \psi)  \sqrt{\lambda_{\max} \nu_{\max} } \] 
\end{lemma}

\begin{proof}
	If $\mathcal{U}$ is a unitary transform using $Q$ qubits of communication, then $\SR(\mathcal{U} \ket{\nu}) \leq 2^Q \SR(\ket{\nu})$  \cite{HW02}.  We also know that the Schmidt coefficients of $\mathcal{U}\ket{\nu}$ are bounded above by $2^Q \nu_{max}$ \cite{HW02}.  The desired result now follows by Lemma \ref{lem:innerprod}.
\end{proof}

\begin{lemma} \label{lem:innerprod}
Given two quantum states $\ket \psi$ and $\ket \nu$ on $\mathcal{H}_A \otimes \mathcal{H}_B$, such that the Schmidt coefficients of $\psi$ are upper bounded by $\lambda_{\max}$, and those of $\nu$ are upper bounded by $\nu_{\max}$,  we have:

\[ |  \braket{\psi}{\nu}| \leq   \SR(\ket \psi)  \sqrt{\lambda_{\max} \nu_{\max} } \]

\end{lemma}

\begin{proof}
For brevity let $r = \SR(\ket \psi)$.  Schmidt decompose $\ket \psi$ and $\ket \nu$ as $\ket \psi = \sum_{i=0}^{r-1} \sqrt{\lambda_i}  \ket i_A \ot \ket i_B$,  as $\ket \nu = \sum_j \sqrt{\nu_j} \ket j_A \ot \ket j_B$.  Define the matrix $M_{\nu} =   \sum_j \sqrt{\nu_j} \ket j_A \ot \bra j_B^*$, and note that

\begin{align*}
\braket{\psi}{\nu} &=  \sum_{i=0}^{r-1} \sum_j \sqrt{\lambda_i \nu_j} \braket{i_A}{j_A} \ot \braket{i_B}{j_B} = \sum_{i=0}^{r-1} \sum_j \sqrt{\lambda_i \nu_j} \braket{i_A}{j_A} \ot (\braket{j_B}{i_B})^* \\
&= \sum_{i=0}^{r-1} \sum_j \sqrt{\lambda_i \nu_j} \braket{i_A}{j_A} \ot \braket{j_B^*}{i_B^*} = \sum_{i=0}^{r-1}\sqrt{\lambda_i } \bra{i_A} \left (\sum_j \sqrt{\nu_j} \ket j_A \ot \bra j_B  \right ) \ket{i_B^*} \\
& = \sum_{i=0}^{r-1} \sqrt{\lambda_i } \bra{i_A} M_{\nu}  \ket{i_B^*} 
\end{align*} 

Now, by definition of a Schmidt Decomposition, we know that the maximum singular value of $M_{\nu}$ is $\sqrt{\nu_{max}}$.  Thus, for all $i$ we have that $ | \bra{i_A} M_{\nu}  \ket{i_B^*} | \leq \sqrt{\nu_{max}}$  (since $\ket{i_A}$ and $\ket{i_B}$ are normalized vectors by definition).  It then follows that:

\begin{align*}
|\braket{\psi}{\nu}| & = \left | \sum_{i=0}^{r-1} \sqrt{\lambda_i } \bra{i_A} M_{\nu}  \ket{i_B^*}  \right |  \leq \sqrt{\lambda_{\max}} \sum_{i=0}^{r-1} | \bra{i_A} M_{\nu}  \ket{i_B^*} | \\
&\leq r \sqrt{\lambda_{\max}\nu_{\max}} = \SR(\ket \psi) \sqrt{\lambda_{\max}\nu_{\max}} 
\end{align*} 

\end{proof}

Theorem \ref{thm::weakplusearth} is the main result of this work. The proof is long enough that a high-level outline may be valuable. Therefore will now give a brief, intuitive outline of the proof of Theorem \ref{thm::weakplusearth}, restated below for the reader's convenience, and include the complete proof in Section \ref{pf::thm::weakplusearth} of the Appendix.  

\begin{theorem*}[Restatement of Theorem \ref{thm::weakplusearth}] Consider a quantum communication protocol $\mathcal{R}$ whose goal is to compute a joint function $g(x,y) \in \{0, 1\}$. Suppose that $\mathcal{R}$ uses an arbitrary bipartite entangled state $\ket{\psi}^{AB}$ (of unbounded dimension), as well as $Q$ qubits of communication total, in either direction (for sufficiently large $Q \geq 15$).  Then, for every $\eps > 0$, there exists a quantum communication protocol $\mathcal{R'}$ which simulates $\mathcal{R}$ with error $\epsilon$, while using only a maximally entangled state as an entangled resource (rather than $\ket{\psi}^{AB}$ or any other state), and using $O(\Qovereps)$ qubits of communication.  Thus, if $\mathcal{R}$ computes $f$ with error $\eps '$ it follows that $\mathcal{R'}$ computes $f$ with error $\epsilon+\eps '$.
\end{theorem*}

\paragraph{Outline of  the Proof of Theorem \ref{thm::weakplusearth}:}The proof of Theorem
\ref{thm::weakplusearth} has three main parts.  
First, the initial entangled state used by the protocol can be converted using a small
amount of communication to a state $\varphi$ in which the Schmidt coefficients are grouped
into evenly spaced groups.  This is achieved using Theorem \ref{thm::earthmover}.

Second, we show that
our new ``grouped" entangled state can be divided into three ``pieces" (more precisely
termed \ssmats\ in Definition \ref{def:subsetmatrix}), one piece which has small trace
norm and can therefore be omitted, one piece called $\vfar$ which only has non-zero terms
which are far from the diagonal in the appropriate basis, and one piece called $\vblock$
which is a block-diagonal mixed state that can be produced with small error and low
communication cost from a maximally entangled state. 

 to show that, if one starts with a quantum communication protocol with an arbitrary
 shared entangled state, then that protocol can be modified, using a small amount of
 additional communication, to instead use an entangled state, $\varphi$, (a property which
 will be useful later in the proof).  Once we have reduced, without loss of generality, to
 appropriately ``grouped" entangled state $\varphi$ in this way, the proof proceeds in two
 halves.  In the first half, which is summed up in Lemma \ref{lem:parta}, we show that 

In the second half of the proof, which is summed up in Lemma \ref{lem:fararesmall}, we show that the $\vfar$ piece of $\varphi$ has very little effect on the outcome of the quantum communication protocol in question.  This means that $\varphi$ can be replaced by $\vblock$ alone while incurring very little error in the outcome of the quantum communication protocol.  Since $\vblock$ can be produced with low cost from a maximally entangled state, this then achieves the desired result.  The full proof of Theorem \ref{thm::weakplusearth} is included in Section \ref{pf::thm::weakplusearth} of the Appendix.
    The role of Lemma \ref{lem:cominnerprod} in the proof is within this step for controlling the terms far from the diagonal, in Lemma \ref{lem:fararesmall}.

\section{The Cost of State Transformation:  A Lower Bound} \label{sec::lowerbound}

It is natural at this point to discuss the background and proof for Theorem \ref{thm::lowerbound}, which establishes a lower-bound on the cost of State Transformation by the $\epsilon$-Smoothed $\ell_{\infty}$ Earth Mover's Distance, and to postpone the discussion of Theorem \ref{thm::earthmover} until Section \ref{sec::Earthmoverstatetransformation}, for two reasons.  First, the proof of Theorem \ref{thm::lowerbound} in this section shares key techniques in common with the proof of Theorem \ref{thm::weakplusearth} in Section \ref{sec::mainresult} above, and so this progression may provide the reader with some continuity of thought while also reiterating the usefulness of the techniques.  Second, Theorem \ref{thm::lowerbound} in this section motivates the notion of the $\ell_{\infty}$ Earth Mover's Distance by highlighting its, perhaps surprising, relevance to \emph{lower} bounding the cost of state transformation.  This prepares the reader with some motivation for why the \emph{upper} bound proven in  Theorem \ref{thm::earthmover}, in  Section \ref{sec::Earthmoverstatetransformation} below, is interesting and potentially useful.  Thus, covering Theorem \ref{thm::lowerbound} at this point may provide the reader with a reason to accept the $\epsilon$-smoothed $\ell_{\infty}$ Earth Mover's Distance as a useful proxy for the cost of State Transformation.

Whereas the proof of Theorem \ref{thm::earthmover} in the next section will make direct use of Definition \ref{def::earthmoverdistance}, the proof of Theorem \ref{thm::lowerbound} in this section is elucidated by first establishing an equivalent formulation of the $\ell_{\infty}$ Earth Mover's Distance which is derived by establishing the relationship between the $\ell_{\infty}$ Earth Mover's Distance as defined in Definition \ref{def::earthmoverdistance}, and the Monge-Kantorovich Transportation distance on the real line, as shown below.  After translating to this equivalent definition, stated in Definition \ref{def:equivEarthMoverDistance}, the generalization to the $\epsilon$-smoothed $\ell_{\infty}$ Earth Mover's Distance in Definition \ref{def:epssmoothdist} is straightforward and natural.

\begin{definition}
	Given two probability distributions $\mu$ and $\nu$ on the real line, and a function $c: \mathbb{R} \times \mathbb{R} \to [0, \infty]$ the corresponding Monge-Kantorovich distance, $d_{MK}(\mu, \nu)$ between $\mu$ and $\nu$ is defined as:
	
	\[  d_{MK}(\mu, \nu) = \inf \left \{  \int_{\mathbb{R} \times \mathbb{R}} c(x,y) d \gamma(x,y) |  \gamma \in \Gamma(\mu, \nu)\right \}. \] 
	
	Where $\Gamma(\mu, \nu)$ is defined to be the collection of all probability distributions on $X \times Y \equiv \mathbb{R} \times \mathbb{R}$ which have marginal on $X$ equal to $\mu$ and marginal on $Y$ equal to $\nu$. 
\end{definition}

In order to translate into a statement about quantum states, we make the following definition in a similar style to Definition \ref{def::earthmoverdistance}:

\begin{definition}\label{def:statedistribution}
	Given a bipartite shared state $\ket \psi = \sum_{i\in X} \sqrt{\psi_i} \ket i \ot \ket i$ let us define a random variable $V_{\psi}$ which takes value $\log(\psi_i)$ with probability $\psi_i$  (note that, since the $\psi_i$ sum to one, this is a well defined random variable).  We now define $p_{\psi}$ to be the probability distribution of this random variable.
\end{definition}

It is clear that, for every $\psi$, $p_{\psi}$ is a probability distribution on the real line.  One may note the following simple relationship between Monge-Kantorovich distance and $\ell_{\infty}$ Earth Mover's Distance:

For any $d > 0$, consider the Monge-Kantorovich distance, $d_{MK}$ where the function $c: \mathbb{R} \times \mathbb{R} \to [0, \infty]$ is defined by $c(x,y) = 1$ if $|x-y| \geq d$ and $c(x,y) = 0$ if $|x-y| < d$. Then, for any two quantum states $\ket \psi$ and $\ket \phi$, we have that $\dearth{\ket \psi}{\ket \phi} < d$ if and only if $d_{MK}(p_{\psi}, p_{\phi}) = 0$.

Given this concrete connection between $\ell_{\infty}$ Earth Mover's Distance and the Monge-Kantorovich distance, we can now make use of the following characterization of Monge-Kantorovich distance for distributions on the real line, which is well known in optimal transport theory:

\begin{fact*} \label{fact::monge}
	Let $\mu$ and $\nu$ be probability distributions supported on the real line, and let $F_{\mu}$ and $F_{\nu}$ be their cumulative distribution functions, respectively.  Then, for any $c: \mathbb{R} \times \mathbb{R} \to [0, \infty]$ :
	\[  d_{MK}(\mu, \nu) \equiv \inf_{\gamma \in \Gamma(\mu, \nu)} \left \{  \int_{\mathbb{R} \times \mathbb{R}} c(x,y) d \gamma(x,y)   \right \} =  \int_0^1 c(F_{\mu}^{-1}(s), F_{\nu}^{-1}(s))ds \]
\end{fact*}

It follows from this Fact, combined with the discussion above, that an equivalent definition of the $\ell_{\infty}$ Earth Mover's Distance is given by:

\begin{definition} \label{def:equivEarthMoverDistance}
	\[ \dearth{\ket \psi}{\ket \phi} \equiv \max_{q \in [0,1]} |F_{p_{\psi}}^{-1}(q)- F_{p_{\phi}}^{-1}(q)|  \]
\end{definition}

In the context of this equivalent formulation of  $\ell_{\infty}$ Earth Mover's Distance, we can succinctly introduce a ``smoothed" version of the same distance.  The reader may note that, since the above definition of $\dearth{\ket \psi}{\ket \phi}$ is evidently not robust against tiny changes of either distribution in the total variation distance it would be impossible to prove a lower bound of the form of Theorem \ref{thm::lowerbound} if stated using that definition.  Hence the motivation for introducing a ``smoothed" version of the distance measure, which has built-in robustness by definition.

\begin{definition}{$\epsilon$-Smoothed $\ell_{\infty}$-Earth Mover's Distance} \label{def:epssmoothdist}
	\[\dearthsmooth{\eps}{\ket \psi}{\ket \phi} \equiv \max_{q \in [0,1]} \min_{r \in [q-\eps, q+\eps]} |F_{p_{\psi}}^{-1}(q)- F_{p_{\phi}}^{-1}(r)| \]

\end{definition}

With this definition in place we can now state the lower bound.

\begin{theorem*}[Restatement of Theorem \ref{thm::lowerbound}]
\ThmLowerBound
\end{theorem*}

Intuitively, Theorem \ref{thm::lowerbound} states that two bipartite shared states which are far apart in the $\epsilon$-Smoothed $\ell_{\infty}$-Earth Mover's Distance, cannot be made equal via a quantum communication protocol unless it uses at least $c \cdot \dearthsmooth{\eps}{\ket \psi}{\ket \phi}$ qubits of communication (for a particular constant $c$ which can be computed from the statement of Theorem \ref{thm::lowerbound}).

\begin{proof}
	Suppose that two bipartite shared states $\ket \psi$ and $\ket \phi$ have $\dearthsmooth{\eps}{\ket \psi}{\ket \phi} = d$.  By definition $\exists p \in [0,1]$ such that 
	
	\begin{align}
		\min_{r \in [p-\eps, p+\eps]} |F_{p_{\psi}}^{-1}(p)- F_{p_{\phi}}^{-1}(r)| = d \label{eq:cumulfunc}
	\end{align}

	Suppose that $F_{p_{\psi}}^{-1}(p)< F_{p_{\phi}}^{-1}(r)$ (if the opposite is true then we simply switch the roles of $\psi$ and $\phi$ and continue with the same proof).  Define $x \equiv F_{p_{\psi}}^{-1}(p)$.  Further define $\ket \psi_{\leq x} \equiv \sum_{\{i : |\log{1/\psi_i}| \leq x \} } \sqrt{\psi_i} \ket i \ot \ket i$, and $\ket \psi_{> x} \equiv \ket \psi - \ket \psi_{\leq x}  $.  Similarly define   $\ket \phi_{\geq x+d} \equiv \sum_{\{i :  |\log{1/\phi_i}| \geq x+d \} } \sqrt{\phi_i} \ket i \ot \ket i$, and $\ket \phi_{< x+d} \equiv \ket \phi -\ket \phi_{\geq x+d}$.  Note that $\ket \psi_{\leq x}$, and $\ket \psi_{> x}$ are orthogonal, as are  $\ket \phi_{< x+d}$ and $\ket \phi_{\geq x+d}$.

    Since we have $x \equiv F_{p_{\psi}}^{-1}(p)$ it follows from the definitions that $|| \ket \psi_{\leq x} ||^2 = p$.  Since $F_{p_{\psi}}(x) = p$, and $F_{p_{\psi}}^{-1}(p)< F_{p_{\phi}}^{-1}(r)$, it follows from Equation \ref{eq:cumulfunc} that $F_{p_{\phi}}(x+d) \leq p-\epsilon$.  Therefore, $|| \ket \phi_{< x+d} ||^2  \leq p-\epsilon$ and thus $|| \ket \phi_{\geq x+d}||^2 = 1-  || \ket \phi_{< x+d} ||^2 \geq 1 - p + \epsilon$.
    
    The main idea in the proof of this theorem is that we can now partition   $\ket \psi, \ket \phi$ into three nearly orthogonal parts, depending on $U_\mathcal{P}$, as follows:

    \begin{definition}\label{def:lowerbounddecomp}
   \begin{equation*}
   \begin{aligned}[c]
   &\ket{\psi^1} \equiv U_\mathcal{P} \ket \psi_{\leq x},\\
  &\ket{\psi^3} \equiv   \ket{\phi^3} \bra{ \phi^3}  U_\mathcal{P} \ket \psi_{> x},
   \end{aligned}
   \qquad
   \begin{aligned}[c]
   &\ket{\phi^3} \equiv \ket \phi_{\geq x+d},\\
   &\ket{\phi^1} \equiv   \ket{\psi^1} \bra{ \psi^1} \ket \phi_{< x+d},
   \end{aligned}
   \qquad
   \begin{aligned}[c]
   &\ket{\psi^2} \equiv \left (I - \ket{\phi^3} \bra{ \phi^3} \right )U_\mathcal{P} \ket \psi_{> x}\\
   &\ket{\phi^2} \equiv  \left (I - \ket{\psi^1} \bra{ \psi^1} \right ) \ket \phi_{< x+d} 
   \end{aligned}
   \end{equation*}
\end{definition}

   \begin{lemma} \label{lem:approxorthog}
   	For $i, j \in \{1,2,3\}$ with $i \neq j$, we have that $|\langle \phi^i|\psi^j \rangle |\leq h(Q,d)$, $|\langle \psi^i|\psi^j \rangle |\leq h(Q,d)$, and $|\langle \phi^i|\phi^j \rangle| \leq h(Q,d)$, where $h(Q,d) \equiv  4 \cdot 2^{\frac{3Q -d}{2}}$.
   \end{lemma}
   
   The proof of Lemma \ref{lem:approxorthog} is given separately in the appendix.  Within that proof is the key use of Lemma \ref{lem:cominnerprod} which is the primary conceptual step in proving Theorem \ref{thm::lowerbound}.  Understanding the proof of Lemma \ref{lem:approxorthog} is also the best way of understanding the motivation behind Definition \ref{def:lowerbounddecomp} above.

It follows from the definitions that:

\begin{align}
& U_\mathcal{P} \ket{\psi} =  \ket{\psi^1} +\ket{\psi^2} + \ket{\psi^3}  \label{eq:normboundpsi} \\ 
& \ket \phi =\ket{\phi^1} + \ket{\phi^2} + \ket{\phi^3} \label{eq:normboundphi}
\end{align} 

While the individual $\ket{\psi^i}$ and $\ket{\phi^i}$ are not necessarily all orthogonal we do have
$\ket{\psi^2}\perp\ket{\psi^3}$ and $\ket{\psi^1}\perp \ket{\psi^2} + \ket{\psi^3}$. 
Likewise
$\ket{\phi^1}\perp\ket{\phi^2}$ and $\ket{\phi^3}\perp \ket{\psi^1} + \ket{\psi^2}$.
Together these imply
\begsub{norm-sum}
1 & =  \left\| \ket{\psi^1} \right\|^2 + \left\| \ket{\psi^2} \right\|^2 + \left\| \ket{\psi^3} \right\|^2\\
1 & = \left\| \ket{\phi^1} \right\|^2 + \left\| \ket{\phi^2} \right\|^2 + \left\| \ket{\phi^3} \right\|^2 
\endsub

  From Lemma \ref{lem:approxorthog} it follows that:

\begin{align}
|\braket{ \phi}{\mathcal{P}(\psi)}| &\equiv |\bra{ \phi}U_\mathcal{P}\ket{\psi}| = \left |\left (\bra{\phi^1} + \bra{\phi^2} + \bra{\phi^3} \right ) \left (\ket{\psi^1} + \ket{\psi^2} + \ket{\psi^3} \right ) \right   |  \nonumber   \\
&\leq \left |\braket{\phi^1}{\psi^1} \right| + \left |\braket{\phi^2}{\psi^2} \right| + \left |\braket{\phi^3}{\psi^3} \right|  + 6 \cdot h(Q,d) \nonumber \\
& \leq \left  \| \ket{\phi^1} \right \| \left \| \ket{\psi^1} \right \| +  \left \| \ket{\phi^2} \right \| \left \| \ket{\psi^2} \right \| +  \left \| \ket{\phi^3} \right \| \left \| \ket{\psi^3} \right \|+ 6 \cdot h(Q,d) \label{eq:threepartnorm}
\end{align}

  Now recall that 

\begin{align}
& \left  \| \ket{\psi^1} \right \|  = \left  \| U_{\mathcal{P}} \ket{\psi}_{\leq x} \right \|  = \left  \| \ket{\psi}_{\leq x} \right \| = \sqrt{p} \nonumber \\
& \left  \| \ket{\phi^3} \right \|  = \left  \|  \ket{\phi}_{\geq x + d} \right \|  \geq \sqrt{1-p+\epsilon} \nonumber
\end{align}

We now return to Equation \ref{eq:threepartnorm}.  Setting $x_i = \|\,|\psi^i\rangle\|$ and
$y_i = \|\,|\phi^i\rangle\|$ for $i=1,2,3$ we have

\begin{align} \label{eq:orthogbound}
|\braket{ \phi}{\mathcal{P}(\psi)}| & \leq
x_1y_1 + x_2y_2 + x_3y_3 + 6 \cdot h(Q,d)
\end{align}
where $x_1=\sqrt p$, $y_3 \geq \sqrt{1-p+\eps}$ and $(x_1,x_2,x_3), (y_1,y_2,y_3)$ are unit vectors. We
claim that this quantity is maximized by setting $x_2=y_2=0$ and $y_3=\sqrt{1-p+\eps}$.  Indeed we can upper bound
$\sqrt{p}y_1 + x_2y_2 \leq x_{12}y_{12}$ where $x_{12}\equiv \sqrt{x_1^2 + x_2^2}$ and
$y_{12}\equiv \sqrt{y_1^2 + y_2^2}$.  Now define $x_{12} = \cos(\alpha),
x_{3}=\sin(\alpha),y_{12}=\cos(\beta),y_3 = \sin(\beta)$ and we have
\be 
x_1y_1 + x_2y_2 + x_3y_3 \leq \cos(\alpha-\beta).\ee
This is maximized by taking $(x_1,x_2,x_3) = (\sqrt p, 0, \sqrt{1-p})$ and
$(y_1,y_2,y_3) = (\sqrt{p-\eps}, 0, \sqrt{1-p+\eps})$.
Thus 
\be |\braket{ \phi}{\mathcal{P}(\psi)}|  \leq
 \sqrt{p- \epsilon} \sqrt{p} + \sqrt{1-p}\sqrt{1-p+\epsilon}  + 6 \cdot h(Q,d).\ee

Finally we would like an upper bound independent of $p$.  This maximization is performed
in the proof of Fact \ref{fct:calculus} from Section \ref{sec:calculus} of the Appendix
and yields the following.

\begin{align} 
|\braket{ \phi}{\mathcal{P}(\psi)}| \leq 1 - \frac{1}{4} \epsilon^2  + 6 \cdot h(Q,d). \nonumber
\end{align}

\end{proof}

\section{The Cost of State Transformation:  An Upper Bound} \label{sec::Earthmoverstatetransformation}

In this section we will give a proof of Theorem \ref{thm::earthmover}, which states that the quantum communication cost of converting between two bipartite entangled states is upper bounded by the $\ell_{\infty}$ Earth Mover's Distance between those states.  This upper bound represents the second half of our two sided argument (employing both Theorem \ref{thm::earthmover} and Theorem \ref{thm::lowerbound}) that the $\ell_{\infty}$ Earth Mover's Distance is a simple and efficiently computable proxy for the cost of state conversion.   The proof is divided into two parts which are proved separately in Lemma \ref{lem::quadpartite}, and  Lemma \ref{lem::rightflowprot} together with Corollary \ref{cor::leftflowprot}.  At a high level Lemma \ref{lem::quadpartite} tells us that, given bipartite states $\ket\chi $ and $\ket \upsilon$, one can map the Schmidt coefficients of  $\ket\chi $ directly onto the Schmidt coefficients of $\ket \upsilon$ using a series of bipartite ``flows" that have small degree  (where degree is a quantity defined below).  Lemma \ref{lem::rightflowprot} and Corollary \ref{cor::leftflowprot} then tell us that any such ``flow" which has small degree, can be implemented as an actual bipartite state transformation, with correspondingly small communication required.

Here we establish Lemmas \ref{lem::quadpartite} and \ref{lem::rightflowprot} which, together, prove the desired theorem.  We begin with a couple definitions establishing the concept of flows, as we use it here.

\begin{definition}[Right (Left) Index-1 Flow ]
	Fix two states $\ket\chi = \sum_{i\in X} \sqrt{\chi_i} \ket i \ot \ket i$ and
        $\ket\upsilon = \sum_{j\in Y} \sqrt{\upsilon_j} \ket j \ot \ket j$. 
A Right Index-1 Flow from $\ket\chi$ to $\ket \ups$ is a bipartite graph $G_{X,Y}$ with vertices given by $X \cup Y$, and edge set $E_{X,Y}$, such that:
	\begin{itemize}
		% \item Each vertex $i \in X$ is assigned weight $w_i  \equiv \chi_i$.
		%\item Each vertex $j \in Y$ is assigned weight $w_j \equiv \ups_j$.
		\item Each vertex in $j \in Y$ has index 1 in $G_{X,Y}$.  
		\item For all $i \in X$, $\chi_i = \sum_{j \in Y: (i,j) \in E_{X,Y}} \ups_j$
		
		If the roles of $\ket \chi$ and $\ket \ups$ are reversed in the above,
                then we say that there is a Left Index-1 Flow from $\ket \ups$ to $\ket
                \chi$. Equivalently, there is a Left Index-1 Flow from $\ket \ups$ to
                $\ket \chi$  exactly when there is a a Right Index-1 Flow from $\ket\chi$
                to $\ket \ups$. 
	\end{itemize}

\end{definition}

\begin{definition}[Degree of a Right (Left) Index-1 Flow ]
	We define the degree of a Right (Left) Index-1 Flow from $\ket\chi = \sum_{i\in X} \sqrt{\chi_i} \ket i \ot \ket i$ to $\ket\upsilon = \sum_{j\in Y} \sqrt{\upsilon_j} \ket j \ot \ket j$ to be the maximum index of any vertex in the bipartite graph $G_{X,Y}$. 
\end{definition}

The following lemma, which a key step in proving Theorem \ref{thm::earthmover},
establishes that bipartite states which are close to each other in the $\ell_{\infty}$
Earth Mover's Distance of Definition \ref{def::earthmoverdistance}, can be mapped to each
other through a series of flows of bounded degree.  This series of flows intuitively
establishes a map for converting one bipartite state to the other using bounded quantum
communication, in a manner that will be made rigorous in Lemma \ref{lem::rightflowprot}.
The main step in the proof of Lemma \ref{lem::quadpartite} involves constructing a flow
through a type of greedy algorithm whose analysis has a number of subtle cases.  In order
to concretely exhibit these cases the entire greedy algorithm, including every case, is
written out in pseudocode in Algorithm 1.

\begin{lemma} \label{lem::quadpartite}
	Given two states $\ket\chi$ and $\ket\upsilon$, there exist two ``intermediate" states $\ket \gamma$ and $\ket \rho$, such that there is a Right Index-1 Flow from $\ket \chi$ to $\ket \gamma$ of degree at most $2^{2\lceil\dearth{\ket \chi}{\ket \upsilon} \rceil  +4}$, a Left Index-1 Flow from $\ket \gamma$ to $\ket \rho$ of degree at most $2^{\lceil\dearth{\ket \chi}{\ket \upsilon} \rceil  +2}$, and a Left Index-1 Flow from $\ket \rho$ to $\ket \ups$ of degree at most $2^{\lceil\dearth{\ket \chi}{\ket \upsilon} \rceil  +2}$.
\end{lemma}

The Proof of Lemma \ref{lem::quadpartite} is included in the Appendix, section \ref{pf::lem::quadpartite}.

Lemma \ref{lem::quadpartite}, above, shows that two bipartite entangled states can be connected to each other by a series of flows which have a degree which is bounded in terms of the $\ell_{\infty}$ Earth Mover's Distance between them.  The next step is to establish that every flow can be implemented via a quantum communication protocol.   Lemma \ref{lem::rightflowprot} and Corollary \ref{cor::leftflowprot}, below, accomplish this by showing that, if two bipartite states can be connected by flows of small degree, then one state can be converted to the other (and vice versa) using a quantum communication protocol which only requires small amounts of communication.

\begin{lemma}  \label{lem::rightflowprot}
	Given two states $\ket\tau$ and $\ket\kappa$ such that there is a Right Index-1 Flow from $\ket \tau$ to $\ket \kappa$ with degree at most $2^Q$, there exists a quantum communication protocol $\mathcal{P}$, which uses $Q$ qubits of communication, and converts the shared state $\ket \tau$ to the shared state $\ket \kappa$.
\end{lemma}

The idea of the proof is that if $\ket\tau=\sum_i \sqrt{\tau_i}\ket{i}\ot \ket i$ then it
suffices to define separately protocols for each $\ket{i}\ot \ket{i}$ term.  These
protocols simply use quantum communication to create a shared entangled state, resulting
in the state $\sum_i \tau_i \ket{i}_A\ot \ket{i}_B \ot \ket{\psi_i}_{A'B'}$.  Choosing the
Schmidt coefficients according to the given Right Index-1 Flow yields the result.  The
details of this argument are in the Appendix~x\ref{pf::lem::rightflowprot}.

Corollary \ref{cor::leftflowprot} establishes the same result as Lemma \ref{lem::rightflowprot}, but in the reverse direction.

\begin{corollary}
	\label{cor::leftflowprot}
	Given two states $\ket\tau$ and $\ket\kappa$ such that there is a Left Index-1 Flow from $\ket \kappa$ to $\ket \tau$ with degree at most $2^Q$, then, for two parties sharing entangled state $\ket \kappa$, there exists a quantum communication protocol $\mathcal{P}$, which uses $Q$ qubits of communication, and converts the shared state $\ket \kappa$ to the shared state $\ket \tau$.
\end{corollary}

The proof of Corollary \ref{cor::leftflowprot} is straightforward and appears in Appendix~\ref{pf::cor::leftflowprot}.

\begin{theorem*} [Restatement of Theorem \ref{thm::earthmover}]   Let $\ket\chi^{AB}$ and $\ket\upsilon^{AB}$ be two bipartite shared states.  There is a protocol $\mathcal{M}_{\chi\rightarrow\upsilon}$ which can prepare $\ket\upsilon$ from $\ket\chi$, using only $4 \lceil \dearth{\ket \chi}{\ket \upsilon} \rceil + 8$ qubits of communication.
\end{theorem*}
\begin{proof}
	The proof follows by applying Lemma \ref{lem::quadpartite}, followed by  Lemma \ref{lem::rightflowprot} and Corollary \ref{cor::leftflowprot}.
\end{proof}

\pagebreak

\appendix

\section{Proof of Lemma \ref{lem:approxorthog}}

\begin{proof}
	First note that it is immediate from the definitions that $\braket{\phi^2}{\psi^1} = \braket{\phi^3}{\psi^2} = 0$, so the conditions of the lemma are automatically satisfied in those cases.
	
	To bound the remaining inner products we will first prove a bound on the inner product $|\braket{\phi^3 }{\psi^1} |$ and note that the remaining inner products are bounded as a consequence of this first bound.  For notational convenience, while establishing the bound on $|\braket{\phi^3 }{\psi^1} |$, we set $\ket \rho \equiv \ket \psi_{\leq x}$, and let $\rho_j$ be the non-zero Schmidt coefficients of $\ket \rho$ (which are just a renamed version of the non-zero Schmidt coefficients of $\ket \psi_{\leq x}$).  Therefore, we know that, for all $j$, $1 \geq \rho_j \geq 2^{-x}$, and $ \ket \psi_{\leq x} = \ket \rho = \sum_j \sqrt{\rho_j} \ket j \ot \ket j$.  The purpose of this renaming convention is that we can now cleanly make the following definition. For integers $i$ define $\ket \rho_i \equiv \sum_{ \{j:  i <  |\log{1/\rho_j}|\leq i+1\}} \sqrt{\rho_j} \ket j \ot \ket j$, so that we have $\ket \psi_{\leq x} = \ket \rho =  \sum_{i = -1}^{\lceil x \rceil} \ket \rho_i$, and $\braket{\rho_k} {\rho_i} = 0$ whenever $k \neq i$. So, 
	\begin{align}
	\sum_{i = -1}^{\lceil x \rceil} \left \| \ket \rho_{i}\right \|^2 = \left \| \ket \rho\right \|^2 \leq 1 \label{eq:sumlessone}
	\end{align}

	By definition, for any $ 1 \leq i \leq \lceil x \rceil$, the Schmidt coefficients of $\ket \rho_{i} $ are upper bounded by $2^{-i}$, and lower bounded by $2^{-(i+1)}$, and from the latter we have $\SR(\ket \rho_{i}) \leq 2^{i+1} \left \|\ket \rho \right \| ^2 $.  Furthermore, the Schmidt coefficients of $\ket{\phi_{\geq x+d}}$ are upper bounded by $2^{-(x+d)}$, and thus, we have by Lemma \ref{lem:cominnerprod}  that:
	
	\begin{align} \label{eq:innerproductbound1}
	&\left|\bra{\phi_{\geq x+d}}U_\mathcal{P}\ket{\rho}_i \right |   \leq     2^{\frac{3}{2}Q} \SR(\ket \rho_{i}) \sqrt{2^{-(x+d)}2^{-i}} \leq 2^{\frac{3}{2}Q} \cdot 2^{i+1}  \left \|\ket \rho_i \right \|^2 \cdot  \sqrt{2^{-(x+d)}2^{-i}} \nonumber\\
	& = 2 \cdot  2^{\frac{3}{2}Q} \left \|\ket \rho_i \right \|^2 \sqrt{2^{i-x-d}} \leq 2 \cdot  2^{\frac{3}{2}Q} \left \|\ket \rho_i \right \|^2 \cdot 2 \cdot 2^{-d/2} = 4 \cdot 2^{\frac{3Q -d}{2}} \left \|\ket \rho_i \right \|^2,   
	\end{align}
	
	where the final inequality follows because  $i \leq \lceil x \rceil$ by assumption.  Thus,

	\begin{align} \label{eq:phi3psi1}
	&\left|\braket{\phi^3 }{\psi^1} \right | = \left|\bra{\phi_{\geq x+d}}U_\mathcal{P}\ket \psi_{\leq x} \right | = \left|\sum_{i = -1}^{\lceil x \rceil}\bra{\phi_{\geq x+d}}U_\mathcal{P}\ket \rho_i \right | \leq  \sum_{i = -1}^{\lceil x \rceil} \left|\bra{\phi_{\geq x+d}}U_\mathcal{P}\ket \rho_i \right | \nonumber\\ 
	& \leq  4 \cdot 2^{\frac{3Q -d}{2}}     \sum_{i = -1}^{\lceil x \rceil} \left \|\ket \rho_i \right \|^2 = 4 \cdot 2^{\frac{3Q -d}{2}} \left \|\ket \psi_{\leq x} \right \|^2 \leq 4 \cdot 2^{\frac{3Q -d}{2}} = h(Q,d),
	\end{align}
	
	where the second inequality follows by Equation \ref{eq:innerproductbound1} and the subsequent equality follows by Equation \ref{eq:sumlessone}.  Having established this upper bound on $\left|\braket{\phi^3 }{\psi^1} \right |$ we now proceed with bounding the other inner products in the Lemma statement:
	
	\begin{align}\label{eq:psi3psi1}
	& \left|\braket{\psi^3 }{\psi^1} \right | = \left|\bra{\psi_{>x}}U_\mathcal{P}^{\dagger} \ket{\phi^3} \braket{ \phi^3} {\psi^1} \right |  =\left|\bra{\psi_{>x}}U_\mathcal{P}^{\dagger} \ket{\phi^3} \right|  \left |\braket{ \phi^3}{\psi^1} \right | \leq  \left |\braket{ \phi^3}{\psi^1} \right |  \leq h(Q,d),
	\end{align}
	
	\begin{align*}
	& \left|\braket{\psi^2 }{\psi^1} \right | = \left|\bra{\psi_{>x}}U_\mathcal{P}^{\dagger}\left (I - \ket{\phi^3} \bra{ \phi^3} \right ) \ket{\psi^1} \right | \leq \left|\bra{\psi_{>x}}U_\mathcal{P}^{\dagger}\ket{\psi^1} \right | +\left|\bra{\psi_{>x}}U_\mathcal{P}^{\dagger} \ket{\phi^3} \braket{ \phi^3}{\psi^1} \right | \\
	& =\left|\bra{\psi_{>x}}U_\mathcal{P}^{\dagger}U_\mathcal{P}\ket{\psi}_{\leq x} \right | + \left|\braket{\psi^3 }{\psi^1} \right | = \left|\braket{\psi_{>x}}{\psi_{\leq x}} \right|  +  \left|\braket{\psi^3 }{\psi^1} \right | = \left|\braket{\psi^3 }{\psi^1} \right | \leq  h(Q,d),
	\end{align*}

	where both of the inequality steps follow by Equation \ref{eq:psi3psi1}  (the first of which also uses the triangle inequality).

	\begin{align*}
	& \left|\braket{\phi^3 }{\phi^1} \right | = \left|\braket{\phi^3}{ \psi^1}\braket{\psi^1}{\phi_{< x+d}} \right |  = \left|\braket{\phi^3}{ \psi^1}\right |  \left|\braket{\psi^1}{\phi_{< x+d}} \right | \leq \left|\braket{\phi^3}{ \psi^1}\right |  \leq h(Q,d),
	\end{align*}
	
	\begin{align*}
	& \left|\braket{\phi^3 }{\phi^2} \right | = \left|\bra{\phi^3}\left (I - \ket{\psi^1} \bra{\psi^1} \right ) \ket{\phi}_{< x+d} \right | \leq \left|\braket{\phi^3}{\phi_{< x+d}} \right | +\left|\braket{\phi^3}{\psi^1} \braket{\psi^1}{\phi_{< x+d}} \right |  \\
	& = \left|\braket{\phi_{> x +d}}{\phi_{< x+d}} \right | +\left|\braket{\phi^3}{\psi^1}\right|  \left| \braket{\psi^1}{\phi_{< x+d}} \right |  = \left|\braket{\phi^3}{\psi^1}\right|  \left| \braket{\psi^1}{\phi_{< x+d}} \right | \leq \left|\braket{\phi^3}{\psi^1}\right| \leq h(Q,d)
	\end{align*}
	
	Now, as noted earlier, $\braket{\phi^2}{\psi^1} = \braket{\phi^3}{\psi^2} = 0$.  Continuing with the cross terms we have:
	
	\begin{align*}
	& \left|\braket{\phi^1 }{\psi^2} \right |  = \left|\braket{\psi^2 }{\phi^1} \right | = \left|\braket{\psi^2}{ \psi^1}\braket{\psi^1}{\phi_{< x+d}} \right |  = \left|\braket{\psi^2}{ \psi^1}\right |  \left|\braket{\psi^1}{\phi_{< x+d}} \right | \leq \left|\braket{\psi^2}{ \psi^1}\right |  \leq h(Q,d),
	\end{align*}
	
	\begin{align*}
	& \left|\braket{\phi^1 }{\psi^3} \right |  = \left|\braket{\psi^3 }{\phi^1} \right | = \left|\braket{\psi^3}{ \psi^1}\braket{\psi^1}{\phi_{< x+d}} \right |  = \left|\braket{\psi^3}{ \psi^1}\right |  \left|\braket{\psi^1}{\phi_{< x+d}} \right | \leq \left|\braket{\psi^3}{ \psi^1}\right |  \leq h(Q,d),
	\end{align*}
	
	where the last inequality follows from Equation \ref{eq:psi3psi1}.  And, since we already have $\left|\braket{\phi^3 }{\psi^1} \right |  \leq h(Q,d)$ from Equation \ref{eq:phi3psi1}, the final inner product to bound is:
	
	\begin{align*}
	& \left|\braket{\phi^2 }{\psi^3} \right |  = \left|\bra{\phi}_{< x+d}\left (I - \ket{\psi^1} \bra{\psi^1} \right )  \ket{\psi^3} \right|\\
	& \leq \left|\braket{\phi_{< x+d} }{\psi^3} \right |  + \left|\braket{\phi_{< x+d}}{\psi^1} \braket{\psi^1}{\psi^3}  \right |\\
	&= \left|\braket{\phi_{< x+d} }{\phi^3} \bra{ \phi^3}  U_\mathcal{P} \ket \psi_{> x} \right |  + \left|\braket{\phi_{< x+d}}{\psi^1} \braket{\psi^1}{\psi^3}  \right |\\
	& = \left|\braket{\phi_{< x+d} }{\phi^3} \right | \left |\bra{ \phi^3}  U_\mathcal{P} \ket \psi_{> x} \right |  + \left|\braket{\phi_{< x+d}}{\psi^1} \right | \left | \braket{\psi^1}{\psi^3}  \right |\\
	&  = \left|\braket{\phi_{< x+d} }{\phi_{> x+d}} \right | \left |\bra{ \phi^3}  U_\mathcal{P} \ket \psi_{> x} \right |  + \left|\braket{\phi_{< x+d}}{\psi^1} \right | \left | \braket{\psi^1}{\psi^3}  \right |\\
	& \leq 0 + \left | \braket{\psi^1}{\psi^3}  \right | \leq h(Q,d),
	\end{align*}

	where the last inequality follows by Equation \ref{eq:psi3psi1}.
	
\end{proof}

\section{Fact \ref{fct:calculus}} \label{sec:calculus}

\begin{fact}\label{fct:calculus}
	For $p \in [0,1]$ and  $0 \leq \epsilon \leq p$,  $\sqrt{p- \epsilon} \sqrt{p} + \sqrt{1-p}\sqrt{1-p+\epsilon}\leq 1 - \frac{1}{8} \epsilon^2$
\end{fact}

\begin{proof}
	Define $f(x) \equiv\sqrt{p- x} \sqrt{p} + \sqrt{1-p}\sqrt{1-p+x}$.  Note that $f'(x) = - \frac{\sqrt{p}}{2 \sqrt{p-x}} + \frac{\sqrt{1-p}}{2\sqrt{1-p+x}}$, and $f''(x) = -1/4 \left (\frac{\sqrt{p}}{(p-x)^{3/2}} + \frac{\sqrt{1-p}}{(1-p+x)^{3/2}} \right )$.  So, $f(0) = 1$, $f'(0) = 0$, and 
	
	$$f''(x) = -1/4 \left (\frac{\sqrt{p}}{(p-x)^{3/2}} + \frac{\sqrt{1-p}}{(1-p+x)^{3/2}} \right ) \leq -1/4 \frac{\sqrt{p}}{(p-x)^{3/2}} \leq -1/4 \frac{1}{p} \leq -1/4 $$
	
	for all $p \in [0,1]$  and $0 \leq x \leq p$.  It follows by integration that:
	
	\[f(x)  = 1+ \int_0^x \int_0^x f''(y) dy dz \leq 1+\int_0^x \int_0^x (-1/4) dy dz = 1-\frac{1}{8} x^2  \]
	
	So, $$\sqrt{p- \epsilon} \sqrt{p} + \sqrt{1-p}\sqrt{1-p+\epsilon} = f(\epsilon)\leq 1 - \frac{1}{8} \epsilon^2$$

\end{proof}

\section{Proof of Lemma \ref{lem::rightflowprot}}\label{pf::lem::rightflowprot}

\begin{proof}
	By assumption there is a Right Index-1 Flow from $\ket \tau$ to $\ket \kappa$ with degree at most $2^Q$, so  there exists a bipartite graph $G_{X,Y}$ with vertices given by $X \cup Y$, and edge set $E_{X,Y}$, such that:
	\begin{itemize}
		% \item Each vertex $i \in X$ is assigned weight $w_i  \equiv \chi_i$.
		%\item Each vertex $j \in Y$ is assigned weight $w_j \equiv \ups_j$.
		\item Each vertex in $j \in Y$ has index 1 in $G_{X,Y}$.  
		\item For all $i \in X$, $\tau_i = \sum_{j \in Y: (i,j) \in E_{X,Y}} \kappa_j$.
		\item  The maximum degree of any vertex $i \in X$ in $G_{X,Y}$ is $2^Q$.  
	\end{itemize}
	
	The protocol for Alice and Bob to start with shared state $\ket \tau$ and end up with shared state $\ket \kappa$ will proceed as follows:  Beginning with the state $\ket \tau$ shared between Alice and Bob, we will refer to the register containing the Alice half of $\ket \tau$ as $A$, and the register containing the Bob half as $B$.   Alice will append two additional registers, of $Q$ qubits each, and initialize each of them to the all zeros state.  We will call these two new registers $C_1$ and $C_2$ respectively.  Alice will then perform a controlled unitary operation between $A$ and the registers $C_1$ and $C_2$.  She will then pass the register $C_2$ to Bob using $Q$ qubits of quantum communication to do so.  Bob will then perform a controlled unitary between $B$ and $C_2$, Alice will perform a controlled unitary between $A$ and $C_1$, and after that Alice and Bob will share the state $\ket \kappa$.
	
	To describe the protocol more precisely we will define the specific controlled unitaries performed by Alice and Bob at each step.  Beginning with a shared state $\ket \tau$, after Alice appends the two additional $Q$-qubit registers to her side of $\ket \tau$, the shared state looks as follows:
	
	$$\ket\tau = \sum_{i\in X} \sqrt{\tau_i} \ket{0^{\ot Q}}_{C_1} \ot \ket{0^{\ot Q}}_{C_2} \ot \ket i_A \ot \ket i_B$$

	Where, initially, Alice holds the registers $A$, $C_1$, and $C_2$.  Alice now performs a controlled unitary operation, acting on registers $C_1$ and $C_2$ and controlled on register $A$.  To describe this controlled unitary concisely we will need to imagine that there is some total order on the elements $j \in Y$  (any total order will do, one can simply imagine that the $j$'s are indexed by bit strings which encode integers), and we will define $s_{ij} \equiv | \{j' \in Y: j' < j, \text{ and } (i,j') \in E_{X,Y}  \}  |  $. Note that, since every $i \in X$ has degree at most $2^Q$, $s_{ij}$ is always an integer between $0$ and $2^Q$, so it can always be expressed in binary as a $Q$-bit binary number.  We will take this convention in the following argument.
	
	Now to define Alice's controlled unitary:  When controlled on $\ket i_A$ Alice's unitary moves the state $\ket{0^{\ot Q}}_{C_1} \ot \ket{0^{\ot Q}}_{C_2}$ to the state $\ket{i\text{-controlled}}_{C_1 C_2} \equiv \sum_{j \in Y: (i,j) \in E_{X,Y}} \sqrt{\kappa_j/\tau_i}\ket{s_{ij}}_{C_1} \ot \ket{s_{ij}}_{C_2}$.  Note that since $s_{ij}$ is always a $Q$-bit binary string, it can always be contained in the $Q$-qubit registers $C_1$ and $C_2$.  Further note that, since $\tau_i = \sum_{j \in Y: (i,j) \in E_{X,Y}} \kappa_j$ by assumption, $\ket{i\text{-controlled}}_{C_1 C_2}$ is a normalized pure state.  Thus there exists a unitary operation that moves $\ket{0^{\ot Q}}_{C_1} \ot \ket{0^{\ot Q}}_{C_2}$ to $\ket{i\text{-controlled}}_{C_1 C_2}$ and Alice need only perform this specific unitary when the control register is in state $\ket i_A$.  So, when Alice applies this controlled unitary to her registers $C_1$, $C_2$ and $A$ (where $A$ is the controlling register), the resulting new shared state between Alice and Bob is:

	\begin{align}
	& \ket\tau = \sum_{i\in X} \ket{i\text{-controlled}}_{C_1 C_2} \ot \ket i_A \ot \ket i_B = \sum_{i\in X} \sum_{j \in Y: (i,j) \in E_{X,Y}} \sqrt{\tau_i} \cdot \sqrt{\kappa_j/\tau_i}\ket{s_{ij}}_{C_1} \ot \ket{s_{ij}}_{C_2} \ot \ket i_A \ot \ket i_B \\
	& = \sum_{i\in X} \sum_{j \in Y: (i,j) \in E_{X,Y}} \sqrt{\kappa_j}\ket{s_{ij}}_{C_1} \ot \ket{s_{ij}}_{C_2} \ot \ket i_A \ot \ket i_B
	\end{align}
	
	At this point Alice uses $Q$ qubits of communication to pass the $Q$-qubit register $C_2$ to Bob.  The resulting shared state is:
	
	$$ \sum_{i\in X} \sum_{j \in Y: (i,j) \in E_{X,Y}} \sqrt{\kappa_j}\ket{s_{ij}}_{C_1}  \ot \ket i_A \ot \ket i_B \ot \ket{s_{ij}}_{C_2} $$
	
	Where Alice owns registers $C_1$ and $A$, and Bob owns registers $C_2$ and $B$.  Now it is not hard to see from the definition of $s_{ij}$ and the fact that every $j \in Y$ has degree exactly 1 in the graph $G_{X,Y}$, that there is a bijection mapping each $j \in Y$ to the tuple $(i, s_{ij})$.  Alice and Bob both know this bijection since they know the description of $G_{X,Y}$, and since bijections are invertible, Alice and Bob can now both apply a local unitary which relabels the basis element $\ket{i} \ot \ket{s_{ij}}$ to the basis element $j$.  The resulting shared state is:
	
	$$ \sum_{i\in X} \sum_{j \in Y: (i,j) \in E_{X,Y}} \sqrt{\kappa_j} \ket j_A \ot \ket j_B = \sum_{j \in Y} \sqrt{\kappa_j} \ket j_A \ot \ket j_B \equiv \ket \kappa $$
	
	Where the first equality follows because each $j \in Y$ appears in the initial sum exactly once (because $j$ has degree exactly one in $G_{X,Y}$).  
	
	This completes the protocol.

\end{proof}

\section{Proof of Corollary \ref{cor::leftflowprot}}\label{pf::cor::leftflowprot}

\begin{proof}
	By definition, if there is a Left Index-1 Flow from $\ket \kappa$ to $\ket \tau$, then there is a  Right Index-1 Flow from $\ket \tau$ to $\ket \kappa$ (which is the starting assumption of Lemma \ref{lem::rightflowprot}).  
	One can check that, in the proof Lemma \ref{lem::rightflowprot}, every operation
        performed by Alice and Bob was reversible.  Therefore, the proof of this corollary
        is simply to start at the end of the proof of Lemma \ref{lem::rightflowprot}, and
        ``reverse" every step of the proof in order from end to beginning  (including the
        communication step...now communication goes from Bob to Alice rather than Alice to
        Bob).  The result is the desired quantum communication protocol, which converts
        the shared state $\ket \kappa$ to the shared state $\ket \tau$ using $Q$ qubits of communication.
\end{proof}

\section{Proof of Theorem \ref{thm::weakplusearth}} \label{pf::thm::weakplusearth}

A concept which will be useful in the proof of Theorem \ref{thm::weakplusearth} is the notion of the spread of a state:

\begin{definition}[Spread]
	For a finite dimensional bipartite entangled state $\ket {\psi}^{AB} = \sum_{i} \sqrt{\psi_i} \ket i^A \ot \ket i^B$ let $\lambda_{max}$ be the maximum of the Schmidt coefficients of $\psi$, and let $\lambda_{min}$ be the minimum Schmidt coefficient.  We define the spread of $\ket \psi$ to be the quantity $\log(\lambda_{max}/\lambda_{min})$.
\end{definition}

We note that the above definition of spread is given in the case of finite dimensional $\ket \psi$, which is the only case we will need.   There is also an $\epsilon$-smoothed variant of the spread of a state~\cite{HW02,Har-spread}, but it will not be needed for this proof.  Within the proof of Theorem \ref{thm::weakplusearth} the spread of a bipartite state will be used as a proxy for the amount of communication required to create that state from a maximally entangled state.  This intuition is formalized, for example, by Theorem \ref{thm::earthmover}, but in this case of converting from a maximally entangled state, is also an implication of earlier works, such as \cite{HL02,HW02}.

\begin{theorem*}[Restatement of Theorem \ref{thm::weakplusearth}] Consider a quantum communication protocol $\mathcal{R}$ whose goal is to compute a joint function $g(x,y) \in \{0, 1\}$. Suppose that $\mathcal{R}$ uses an arbitrary bipartite entangled state $\ket{\psi}^{AB}$ (of unbounded dimension), as well as $Q$ qubits of communication total, in either direction (for sufficiently large $Q \geq 15$).  Then, for every $\eps > 0$, there exists a quantum communication protocol $\mathcal{R'}$ which simulates $\mathcal{R}$ with error $\epsilon$, while using only a maximally entangled state as an entangled resource (rather than $\ket{\psi}^{AB}$ or any other state), and using $O(\Qovereps)$ qubits of communication.  Thus, if $\mathcal{R}$ computes $f$ with error $\eps '$ it follows that $\mathcal{R'}$ computes $f$ with error $\epsilon+\eps '$.
\end{theorem*}
\begin{proof}
	Given $\mathcal{R}$, $g$, and $\ket \psi$ as in the theorem statement, Schmidt decompose $\ket{\psi}$ as $\sum_i \sqrt{\lambda_i} \ket{i,i}$  (see Remark \ref{rem:decomp} for why we may assume WLOG that $\ket \psi$ has this form).   
	
	Let $N\geq 2$ be an integer, which will be specified later.  Define a function $f: [0,1]
        \to \{0,1,\ldots,N\}$ given by 
	
	\[f(\lambda) =   2^{\left \lceil  \left \lceil \frac{\log(1/\lambda)}{N}\right
              \rceil N - \log(1/\lambda) \right \rceil} 
\in \{1,2,4,\ldots,2^N\}, \]

	and define a new state $\ket{\varphi} \equiv \sum_i  \sum_{j \in \{  1, ..., f(\lambda_i)\}} \sqrt{\nu_{i,j}} \ket{(i,j),(i,j)}$, where $\nu_{i,j} \equiv \frac{\lambda_i}{f(\lambda_i)}$.  Note that $\sum_{i,j} \nu_{i,j} = 1$, so that $\ket{\varphi}$ is a normalized pure state.  Furthermore, every Schmidt coefficient $\nu_{i,j}$ of $\ket{\varphi}$ is within a multiple of $2$ of the integer power $2^{- \left \lceil \frac{\log(1/\lambda_i)}{N}\right \rceil N}$.  This follows because 
	
	\begin{align}
	 \left | \log \left ( \frac{\nu_{i,j}}{2^{- \left \lceil \frac{\log(1/\lambda_i)}{N}\right \rceil N}} \right )  \right |& = \left |  \log \left (\lambda_i \right) - \log(f(\lambda_i)) + \left \lceil \frac{\log(1/\lambda_i)}{N}\right \rceil N \right | \nonumber \\
	&  = \left |  \log \left (\lambda_i \right) -  \log \left (   2^{ \left \lceil  \left \lceil \frac{\log(1/\lambda_i)}{N}\right \rceil N - \log(1/\lambda_i) \right \rceil} \right ) + \left \lceil \frac{\log(1/\lambda_i)}{N}\right \rceil N \right |  \nonumber \\
	&  =  \left |\left \lceil \frac{\log(1/\lambda_i)}{N}\right \rceil N -  \log \left (1/\lambda_i \right) -  \left \lceil  \left \lceil \frac{\log(1/\lambda_i)}{N}\right \rceil N - \log(1/\lambda_i) \right \rceil \right | \nonumber\\
	&  \leq 1 \label{eq:spreadcalc}
	\end{align}

Next, we can upper bound  $\dearth{\ket \psi}{\ket \varphi}\leq N$ by considering the coupling
in which each $\nu_{i,j}$ is moved to $\lambda_i$.  The largest distance obtained here is
the maximum $\log f(\lambda_i)$ for which $\lambda_i>0$, and this in turn is $\leq N$.
Therefore, by Theorem \ref{thm::earthmover}, there is a protocol $\mathcal{M}$ by which
Alice and Bob can prepare $\ket\varphi$ from $\ket\psi$, using $4 \lceil \dearth{\ket
  \chi}{\ket \upsilon} \rceil + 8 \leq 4 N + 8$ qubits of communication.  (For this
special case, of course a simpler protocol could also be used.)
	
	Define $\mathcal{C} \equiv \mathcal{R} \circ \mathcal{M}$ to be the composed protocol in which Alice and Bob start with shared state $\ket\varphi$, first use protocol $\mathcal{M}$ to convert $\ket\varphi$ to $\ket \psi$, and then perform protocol $\mathcal{R}$ using shared state $\ket \psi$ and inputs $x$ and $y$, to compute the joint function $g(x,y)$.  It is evident that $\mathcal{C}$ has exactly the same success probability as $\mathcal{R}$. Since $\mathcal{M}$ uses at most $4N +8$ qubits of communication and $\mathcal{R}$ uses $Q$ qubits of communication, $\mathcal{C}$ can be performed with $Q+4N+8$ qubits of communication.

For $j$ a nonnegative integer, define  $I_j := \{i: 2^{-j N +1} \geq \lambda_i > 2^{- j N
  - 1} \}$ and define the subnormalized state	
	\begin{align}
	\ket{\varphi_j} \equiv \sum_{i \in I_j}  \sqrt{\lambda_i} \ket{i,i}\label{eq:defphij}.
	\end{align}
From Equation \eqref{eq:spreadcalc} and the surrounding discussion, we have that
$\ket{\varphi} = \sum_j \ket{\varphi_j}$.  Furthermore, by the definition of $I_j$, it
follows that $\ket{\varphi_j}$ has spread at most $2$; note that the spread of
$\ket{\varphi_j}$ does not depend on whether the state is normalized or not.
	
	The idea of the proof is that different $\ket{\varphi_j}$ are not only orthogonal, but must remain approximately orthogonal even after a small amount of quantum communication. In particular, note that for any $j$, $\SR(\ket{\varphi_j}) \leq  2^{jN+1} \|\ket{\varphi_j} \|^2$. Furthermore, for all $l$ we have, by definition, that the Schmidt coefficients of $\ket{\varphi_l}$ are bounded above by $2^{-lN+1}$.  Therefore, if $U$ is a unitary transform using $M$ qubits of communication, then, it follows by Lemma \ref{lem:cominnerprod}, that $\forall j,k$,
	
\begin{align}
\left |\bra{\varphi_k} U \ket{\varphi_j} \right | & \leq 2^{\frac{3}{2}M} 2^{\min(j,k)N+1} \left \|\ket{\varphi_{\min(j,k)}} \right \|^2 \sqrt{2^{-jN+1} \cdot 2^{-kN+1}} \nonumber \\
& \leq 2^{\frac{3}{2}M} 2^{-N\frac{|j-k|}{2}+2} \left \|\ket{\varphi_{\min(j,k)}} \right \|^2 \label{eq:innerprod}
\end{align}

To apply this to our problem, we first note that the protocol $\mathcal{C}$ depends, a
priori, on the inputs $x,y$ to the function $g(x,y)$ that we wish to compute (just like
the the protocol $\mathcal{R}$).  We now fix any input pair $x,y$ and for the remainder of
the proof of this theorem we will perform only transformations of the shared state which
do not depend on the value of $x,y$.  We will therefore establish that our transformation
to a maximally entangled shared state does not significantly impact the success
probability of the quantum communication protocol \emph{regardless} of the value of $x,y$.
The desired Theorem then follows.

With the input $x,y$ now fixed, we observe that the success probability of protocol
$\mathcal{C}$ (which we have already established is equal to the success probability of
the original protocol $\mathcal{R}$) can be expressed WLOG by performing $\mathcal{C}$ and
then computing the probability of outcomes when measuring the first qubit in the
computational basis.  The probability that such a measurement on protocol $\mathcal{C}$
outputs $b \in \{0,1\}$ is
	
	\[\Pr[b]  =\bra{\varphi} \mathcal{C}^{\dagger} (\proj{b} \ot I) \mathcal{C} \ket{\varphi},\]

	where $I$ acts on all qubits except for the first, which is being measured. Define $\mathcal{P} \equiv \mathcal{C}^{\dagger} (\sigma_z \ot I) \mathcal{C} = \mathcal{C}^{\dagger} (\ket{0}\bra{0} \ot I) \mathcal{C} - \mathcal{C}^{\dagger} (\ket{1}\bra{1} \ot I) \mathcal{C} $. Then
	
	\begin{equation}
	\Pr[0] - \Pr[1] = \bra{\varphi}\mathcal{P}\ket{\varphi} = \sum_{j,k} \bra{\varphi_j} \mathcal{P}  \ket{\varphi_k} \label{eq:probdiff}
	\end{equation}
	
	Observe, for later, that $\mathcal{P}$ is a unitary operator that can be implemented using $2Q+8N + 16$ qubits of
	communication.  
	
 The proof will proceed as follows:  In Lemma \ref{lem:parta} we show that the density matrix $\varphi = \ket \varphi \bra \varphi$ can be divided into three ``pieces" (in a manner that does not depend on the inputs $x,y$), one piece which has small trace norm and can therefore be omitted, one piece called $\vfar$ which only has non-zero terms which are far from the diagonal in the appropriate basis, and one piece called $\vblock$ which is a block-diagonal mixed state that can be produced with small error and low communication cost from a maximally entangled state.  Then, in Lemma \ref{lem:fararesmall}, we show that the $\vfar$ piece of $\varphi$ has very little effect on the protocol $\mathcal{C}$.  This means that $\varphi$ can be replaced by $\vblock$ alone while incurring very little error in the success probability of $\mathcal{C}$.  Stated equivalently, via the equality in Equation \ref{eq:probdiff} above, Lemma \ref{lem:fararesmall} shows that the quantity $\left | \Tr(\mathcal{P} (\varphi - \vblock)) \right |$ is small.  Since we know from Lemma \ref{lem:parta} that $\vblock$ can be produced with low cost from a maximally entangled state, this leads us to the desired result. Since $\vblock$ does not depend on the inputs $x,y$ this same statement holds for every pair of inputs $x,y$.  From this point forward we will no longer specify the fixed inputs $x,y$, as it will be clear that the state substitutions do not depend on these inputs, and thus that the argument holds for every input as discussed in this paragraph.	
	
    We now establish some notation which will be useful throughout the rest of the proof:

    	\begin{definition}[subset-matrix] \label{def:subsetmatrix}
    	Consider operators on the Hilbert space which is the span of the $\ket{\varphi_j}$.  We say that an operator $M'$ is a \emph{\ssmat} of an operator $M$, if it is the case that for all $l,k$ either   $\bra{\varphi_l}M'\ket{\varphi_k} = \bra{\varphi_l}M\ket{\varphi_k}$, or $\bra{\varphi_l}M'\ket{\varphi_k} = 0$.
    	\end{definition} 
    	
    	\begin{definition}[Non-Zero Set] \label{def:nonzeroset}
    		For an operator $\theta$ on the Hilbert space which is the span of the $\ket{\varphi_j}$, define the non-zero set of $\theta$ to be $T_{\theta} = \{(l,k): \bra{\varphi_k} \theta \ket{\varphi_l} \neq 0\}$.  
    	\end{definition}

\begin{lemma} \label{lem:parta}
Consider the density matrix $\varphi \equiv  \sum_{k,l}  \ket{\varphi_k}\bra{\varphi_l}$.
For any $\epsilon >0$, there exist \ssmats, $\vblock, \vfar$, of $\varphi$,  such that
\benum
\item 
$\|\varphi - (\vblock+\vfar) \|_1 \leq 2 \epsilon$
\item  $T_{\vfar} \subseteq \{(l,k): |k-l| > B\}  $, 
where $B \equiv 30 + 2 \left \lceil \frac{\log(1/\epsilon)}{N}\right \rceil$.
\item The bipartite
shared state $\vblock$ can be prepared starting from EPR pairs with
$O(N/\epsilon+\log(1/\epsilon)/\epsilon)$ bits of communication. 
\eenum
\end{lemma}
	
    The proof of Lemma \ref{lem:parta} is included in Section \ref{sec:parta} of the Appendix.

		We can now bound the difference between the protocol $\mathcal{C}$ acting on $\varphi$ versus $\mathcal{C}$ acting on $\vblock$, following equation \ref{eq:probdiff} as follows:

		\begin{align*}
		&\left | (\text{Pr}_{\varphi}[0] - \text{Pr}_{\varphi}[1]) - (\text{Pr}_{\vblock}[0] - \text{Pr}_{\vblock}[1])   \right| = \left | \Tr(\mathcal{P} (\varphi - \vblock)) \right | 
		\end{align*}
	
	Setting $N = 2Q$ and recalling from the Theorem statement that $Q \geq 15$ by assumption, it follows by Lemma \ref{lem:fararesmall}, stated below, that:
	
	\begin{align}
	&\left | (\text{Pr}_{\varphi}[0] - \text{Pr}_{\varphi}[1]) - (\text{Pr}_{\vblock}[0] - \text{Pr}_{\vblock}[1])   \right|= \left | \Tr(\mathcal{P} (\varphi - \vblock)) \right | \leq 3 \epsilon \label{eq:traceprobdiff}
	\end{align}
	
	This completes the proof of the Theorem as we now describe.
	
	We know from Lemma \ref{lem:parta} that there is a quantum communication protocol, call it $\mathcal{K}$, which prepares the shared state $\vblock$ starting from just a maximally entangled state using at most $O(N/\epsilon+\log(1/\epsilon)/\epsilon)$ bits of communication.  Now define the protocol $\mathcal{R'} \equiv \mathcal{C} \circ \mathcal{K}$.  Since $\mathcal{C}$ uses at most $Q+4N+8$ qubits of communication, and since we have chosen to set $N = 2Q$ (in the line above Equation \ref{eq:traceprobdiff}), it follows that $\mathcal{R'}$ uses at most $O(N/\epsilon+\log(1/\epsilon)/\epsilon) = O(Q/\epsilon+\log(1/\epsilon)/\epsilon)$ qubits of communication.  Furthermore, the success probability of $\mathcal{R'}$ with only the maximally entangled state as an entangled resource is the same, by construction, as the success probability of $\mathcal{C}$ with $\vblock$ as an entangled resource, which, by Equation \ref{eq:traceprobdiff} above and the original definition $\mathcal{C} \equiv \mathcal{R} \circ \mathcal{M}$, is within $3 \epsilon$ of the success probability of the original protocol $\mathcal{R}$ from the theorem statement when using the original shared state $\ket{\psi}$ as an entangled resource.  This is the desired result.	
\end{proof}
	
	\begin{lemma} \label{lem:fararesmall}
		For $\vblock$ as constructed in Lemma \ref{lem:parta}, and for $N, Q$ as defined in the proof of Theorem \ref{thm::weakplusearth} we have, $\left | \Tr(\mathcal{P} (\varphi - \vblock)) \right | \leq 3 \epsilon$ whenever $N \geq 2Q \geq 30$.
	\end{lemma}

\begin{proof}

Following Lemma \ref{lem:parta}, we define $\B \equiv 30 + 2 \left \lceil
  \frac{\log(1/\epsilon)}{N} \right \rceil$.  Now, letting $\vblock$ and $\vfar$ be as in
Lemma \ref{lem:parta}, and recalling that $\|\varphi - (\vblock + \vfar)\|_1 \leq 2
\epsilon$, we have:

	\begin{align*}
	 \left | \Tr(\mathcal{P} (\varphi - \vblock)) \right | &\leq \left | \Tr(\mathcal{P} ((\vblock + \vfar) -\vblock)) \right | + 2 \epsilon = \left | \Tr(\mathcal{P} \vfar) \right | + 2 \epsilon \\
	&= \left |  \sum_{(k,l) \in T_{\vfar} }  \bra{\varphi_k} \mathcal{P}  \ket{\varphi_l} \right |+ 2 \epsilon \leq   \sum_{(k,l) \in T_{\vfar}}  \left |\bra{\varphi_k} \mathcal{P}   \ket{\varphi_l} \right |+ 2 \epsilon  \\
	& \leq \sum_{k,l:  |k-l| > \B} \left |\bra{\varphi_k} \mathcal{P}   \ket{\varphi_l} \right | + 2 \epsilon
	\end{align*}

 where the final inequality follows because $T_{\vfar} \subseteq \{(l,k): |k-l| > \B  \}$ by Lemma \ref{lem:parta}.  Recalling that the unitary $\mathcal{P}$ can be implemented using 2Q+8N+16 qubits of
	communication, and applying equation \ref{eq:innerprod} then gives that:
	
	\begin{align*}
	\left | \Tr(\mathcal{P} (\varphi - \vblock)) \right | - 2 \epsilon & \leq  \sum_{k,l:  |k-l| > \B} \min( 1, 2^{3/2 \cdot (2Q+8N+16) } 2^{-N\frac{|k-l|}{2}+4}) \left \|\ket{\varphi_{\min(k,l)}} \right \|^2  \\
	&= 2 \sum_{l} \left \|\ket{\varphi_{l}} \right \|^2  \sum_{k > l + \B}  \min( 1, 2^{\keyexponent} 2^{-N\frac{|k-l|}{2}+4})  \\
	& = 2   \sum_{n > \B}  \min( 1, 2^{\keyexponent} 2^{-N\frac{n}{2}+4}) \\
	&\leq  2 \cdot   2^{\keyexponent} 2^{-\B N/2+4 } \sum_{k =0}^{\infty}    2^{-N\frac{k}{2}}   \\
	& =2 \cdot    2^{\keyexponent} 2^{-\B N/2+4}  \left (1+\frac{2^{-\frac{N}{2}}}{1-2^{-\frac{N}{2}}} \right )  \\
	& \leq 4 \cdot 2^{\keyexponent} 2^{-\B N/2+4} 
	\end{align*}
	
	So, recalling from the Lemma statement that $N \geq 2Q \geq 30$ by assumption:

	\begin{align*}
	\left | \Tr(\mathcal{P} (\varphi - \vblock)) \right | - 2 \epsilon & \leq 4 \cdot 2^{\keyexponent} 2^{-\B N/2+4} \\
	& \leq  4 \cdot 2^{28} 2^{14N} 2^{-15N- \left \lceil \frac{\log(1/\epsilon)}{N} \right \rceil N}\\
	& \leq  2^{30} 2^{-N- \log(1/\epsilon)} \\
	&\leq  \epsilon
	\end{align*}
	
	So, 
	
	\[\left | \Tr(\mathcal{P} (\varphi - \vblock)) \right | \leq 3 \epsilon \]

	\end{proof}

	Note, in the pre-processing step in the proof of Theorem \ref{thm::weakplusearth}, and again at a point within the proof of Lemma \ref{lem:parta} we use our Theorem \ref{thm::earthmover} in a setting where either the starting or ending state is very close to a maximally entangled state.  It is helpful to observe, to avoid confusion, that in such cases Theorem \ref{thm::earthmover} is not strictly necessary and could be replaced with previously known results from, for example, \cite{HL02,HW02}.  In this manuscript we will use Theorem \ref{thm::earthmover} in these cases in order to remain self-contained, and for the convenience of the reader, but we emphasize that the lines of the proof of Theorem \ref{thm::weakplusearth} in which we use Theorem \ref{thm::earthmover} could be replaced with known results.

\section{Proof of Lemma \ref{lem:parta}} \label{sec:parta}

	\begin{lemma*} [Restatement of Lemma \ref{lem:parta}]
Consider the density matrix $\varphi \equiv  \sum_{k,l}  \ket{\varphi_k}\bra{\varphi_l}$.
For any $\epsilon >0$, there exist \ssmats, $\vblock, \vfar$, of $\varphi$,  such that
\benum
\item 
$\|\varphi - (\vblock+\vfar) \|_1 \leq 2 \epsilon$
\item  $T_{\vfar} \subseteq \{(l,k): |k-l| > B\}  $, 
where $B \equiv 30 + 2 \left \lceil \frac{\log(1/\epsilon)}{N}\right \rceil$.
\item The bipartite
shared state $\vblock$ can be prepared starting from EPR pairs with
$O(N/\epsilon+\log(1/\epsilon)/\epsilon)$ bits of communication. 
\eenum
	\end{lemma*}

\begin{proof}
	Note:  The terminology used in this proof is defined in the proof of Theorem \ref{thm::weakplusearth} preceding the use of Lemma \ref{lem:parta} there (Appendix \ref{pf::thm::weakplusearth}).

	Fixing an $\epsilon > 0$ we will now show how to ``cut" $\varphi \equiv  \sum_{k,l}  \ket{\varphi_k}\bra{\varphi_l}$ down into a mixture of states of small spread such that the cut only removes \ssmats\ of the operator which are either far from the diagonal or small in the trace norm (less than $2\epsilon$).

Define a sequence of mutually orthogonal projectors $\{P_i\}$, where each $P_i$ is the
projection onto the span of $\{\ket{\varphi_l} \}_{2(i-1)B<l\leq 2i \cdot B }$.
Let $$M_i \equiv (P_{2i-1}+P_{2i})  \varphi (P_{2i-1}+P_{2i}).$$

Now, for $k \in [1, ...., \lceil 1/\eps \rceil ]$ define
$$S_k \equiv \sum_{i= 0}^{\infty}M_{i \cdot \lceil 1/\eps \rceil + k}.$$
  The $S_k$ are
block-diagonal \ssmats\ of $\varphi$, which are disjoint in the sense that
$T_{S_k}\cap T_{S_{k'}} = \emptyset$ when $k \neq k'$.  Additionally,
$\sum_{k=1}^{\lceil 1/\eps \rceil } S_k = \sum_i M_i$ is a \ssmat\ of $\varphi$ which contains the
entire diagonal of $\varphi$.  Indeed $\sum_{k=1}^{\lceil 1/\eps \rceil } S_k$ can be
obtained from $\varphi$ via the ``pinching'' TPCP which has Kraus operators given by the
$\{P_{2i-1}+P_{2i}\}$.  Thus
$$1 = \tr \sum_{k=1}^{\lceil 1/\eps \rceil } S_k.$$
Choose $k'$ such that $\tr [S_{k'}] \leq 1/\lceil 1/\eps \rceil \leq \eps$. Since the
$S_k$ are all PSD we also have $\|S_{k'}\|_1\leq \eps$.

Our strategy now is to use something like $\varphi - S_{k'}$ as a candidate for $\vblock + \vfar$ in the Lemma statement.  However, subtracting all of $S_{k'}$ removes some terms close to the diagonal, which, even though it is not a large fraction of all entries in $\varphi$, would make the proof and statement of Lemma \ref{lem:parta} somewhat awkward.  So, in order to make the Lemma statement as clean as possible we will only subtract the ``anti-diagonal" parts of $S_{k'}$, and leave the ``diagonal" parts of $S_{k'}$ in a manner made precise below.

Define the block matrices 
\begin{align}
D_i & \equiv P_{2i-1}\varphi P_{2i-1} + P_{2i}\varphi P_{2i}
A_i  \equiv P_{2i-1}\varphi P_{2i} + P_{2i}\varphi P_{2i-1}
\end{align}
$D_i$ and $A_i$ are, respectively, the diagonal and off-diagonal blocks of $M_i$.

%and the ``off-diagonal'' matrix
%\equiv P_{2i-2}^cP_{2i-1}M_i P_{2i-2}^cP_{2i-1} + P_{2i-1}^cP_{2i} M_i P_{2i-1}^cP_{2i}$,
%which consists of the upper left hand corner and bottom right hand corner of $M_i$.
%Define $A_i \equiv M_i - D_i$, which is the matrix containing just the upper right hand
%corner, and bottom left hand corner of $M_i$.  
Further define $K_{k'} \equiv \sum_{i=
  0}^{\infty}A_{i \cdot \lceil 1/\eps \rceil + k'}$.  We have that $K_{k'} = S_{k'} -
\sum_{i= 0}^{\infty}D_{i \cdot \lceil 1/\eps \rceil + k'}$, and that $\|\sum_{i=
  0}^{\infty}D_{i \cdot \lceil 1/\eps \rceil + k'}\|_1 = \|S_{k'} \|_1$ since $\sum_{i=
  0}^{\infty}D_{i \cdot \lceil 1/\eps \rceil + k'}$ is a block-diagonal \ssmat\ of
$S_{k'}$ containing the entire diagonal of $S_{k'}$.  Thus, 
	
	\[\|K_{k'}\|_1 = \|S_{k'} - \sum_{i= 0}^{\infty}D_{i \cdot \lceil 1/\eps \rceil + k'}\|_1 \leq \|S_{k'}\|_1 + \|\sum_{i= 0}^{\infty}D_{i \cdot \lceil 1/\eps \rceil + k'}\|_1  = 2 \|S_{k'}\|_1 \leq 2 \epsilon \]

	We now define a ``cut down" version of $\varphi$ by $\tilde{\varphi} \equiv \varphi - K_{k'}$.  From this definition we have: 
	
	\begin{align}
	\|\varphi -  \tilde{\varphi} \|_1 = \| K_{k'} \|_1 \leq 2 \epsilon.\label{eq:smalltracenormvar}
	\end{align}

Further, we define the projectors 
	
	\begin{align}
%	P_j \equiv P_{2((j-1) \cdot \lceil 1/\eps \rceil + k' ) - 1}^c P_{2(j \cdot
%	\lceil 1/\eps \rceil + k') - 1}, \label{eq:defpij}
\Q_j \equiv 
\sum_{2((j-1) \cdot \lceil 1/\eps \rceil + k') \leq l < 2(j \cdot \lceil 1/\eps \rceil +
          k') } P_l, \label{eq:defpij}
	\end{align}

	and define the block diagonal matrix $\vblock$ as:
	
	\begin{align}
	\vblock \equiv \sum_j Q_j \tilde{\varphi}  Q_j = \sum_j Q_j (\varphi - K_{k'})  Q_j = \sum_j Q_j \varphi  Q_j \label{eq:defphiprime}
	\end{align}

	where the last equality follows because $\sum_j Q_j  K_{k'} Q_j = 0$ because $K_{k'}$ consists only of the ``anti-diagonal" components $A_{i \cdot \lceil 1/\eps \rceil + k'}$ which lie outside of the $Q_j$. Note that $\vblock$ is a \ssmat\ of $\tilde{\varphi}$ according to Definition \ref{def:subsetmatrix}.  Now define $\vfar$ by:
	
	\begin{align}
	\vfar \equiv \tilde{\varphi} - \vblock \label{eq:defvfar}
	\end{align} 
	
	Therefore, $\vfar$ is also a \ssmat\ of $\tilde{\varphi}$ according to Definition \ref{def:subsetmatrix}.  Furthermore, it follows immediately using Equation \ref{eq:smalltracenormvar} that:
	
	\begin{align}
	\|\varphi - (\vfar+\vblock) \|_1 = \|\varphi -  \tilde{\varphi} \|_1 \leq 2 \epsilon \label{eq:smalltracenorm}
	\end{align}

	\textbf{Second Claim:}  To establish the second claim in Lemma \ref{lem:parta} we
        now show that $T_{\vfar} \subseteq \{(l,k): |k-l| > \B \}  $ (recall that $\B
        \equiv 30 + 2 \left \lceil \frac{\log(1/\epsilon)}{N} \right \rceil$).  To see this, we consider the case that $|k-l| \leq B$ and show that in this
        case $(l,k) \notin T_{\vfar}$.  Assume WLOG that $k \geq l$.  When $|k-l| \leq B$ we know that either $ \exists j$ such that:
	
	\begin{align}\label{eq:inblock}
	2B(2(j-1)\lceil 1/\epsilon \rceil + 2 k' - 1) < l,k \leq 2B(2j\lceil 1/\epsilon \rceil + 2 k' - 1)
	\end{align}

	or $\exists j$ such that:
	
	\begin{align}\label{eq:inK}
	4B(j\lceil 1/\epsilon \rceil +  k') - 3B \leq l \leq 2B(2j\lceil 1/\epsilon \rceil + 2 k' - 1) \leq k \leq 4B(j\lceil 1/\epsilon \rceil +  k') - B
	\end{align}
	
	In the first case, denoted by Equation \ref{eq:inblock}, we have that the coordinates $(l,k)$ lie within the \ssmat\ $\vblock$ of $\varphi$, and thus that either $(l,k) \in T_{\vblock}$  or $(l,k) \notin T_{\varphi}$ by definition.  In particular, either $(l,k) \in T_{Q_j \varphi  Q_j} \subseteq T_{\vblock}$ as follows by Equation \ref{eq:defphiprime} and the definition of $Q_j$ in Equation \ref{eq:defpij}, or $(l,k) \notin T_{\varphi}$.  If  $(l,k) \in T_{\vblock}$ then we note that $T_{\vblock} \cap T_{\vfar} = \emptyset$ by definition (Equation \ref{eq:defvfar}), and this implies that $(l,k) \notin T_{\vfar}$.  If $(l,k) \notin T_{\varphi}$, then $(l,k) \notin T_{\vfar}$ because $T_{\vfar} \subseteq T_{\varphi}$.
	
	On the other hand, in the case denoted by Equation \ref{eq:inK}, we have the coordinates $(l,k)$ lie within the \ssmat\ $K_{k'}$ of $\varphi$, and thus that either $(l,k) \in T_{K_{k'}}$, or $(l,k) \notin T_{\varphi}$. The reason for this is that we know that, in this case, the coordinates $(l,k)$ are within the \ssmat\ $M_{j \lceil 1/\epsilon \rceil + k'}$ of $\varphi$.  Furthermore, since we have already ruled out the case of Equation \ref{eq:inblock}, we know that $(l,k)$ is not in $D_{j \lceil 1/\epsilon \rceil + k'}$, the block diagonal portion of $M_{j \lceil 1/\epsilon \rceil + k'}$.  Therefore, the coordinates $(l,k)$ must lie in the block-anti-diagonal portion $A_{j \lceil 1/\epsilon \rceil + k'} = M_{j \lceil 1/\epsilon \rceil + k'} - D_{j \lceil 1/\epsilon \rceil + k'}$ (this can also be determined directly from Equation \ref{eq:inK} itself, and the definition of $A_{j \lceil 1/\epsilon \rceil + k'}$).  Since $K_{k'} \equiv \sum_{i=0}^{\infty} A_{i \lceil 1/\epsilon \rceil + k'}$ we know that the coordinates $(l,k)$ lie within the $K_{k'}$, or more precisely, either $(l,k) \in T_{K_{k'}}$, or $(l,k) \notin T_{\varphi}$.  Just as before, if $(l,k) \notin T_{\varphi}$, then $(l,k) \notin T_{\vfar} \subseteq T_{\varphi}$.  On the other hand, in the case that $(l,k) \in T_{K_{k'}}$ we know that $T_{K_{k'}} \cap T_{\vfar} = \emptyset$ because $T_{\vfar} \subseteq T_{\tilde{\varphi}}$ by Equation \ref{eq:defvfar}, and $T_{\tilde{\varphi}} \cap T_{K_{k'}} = \emptyset$ as follows from the definition $\tilde{\varphi} \equiv \varphi - K_{k'}$.  
	
    This establishes that $T_{\vfar} \subseteq \{(l,k): |k-l| > \B \} $.

	\textbf{Third Claim:}  To establish the third claim in Lemma \ref{lem:parta}, and complete the proof, we will show that $\vblock$ is a mixture of states of spread at most $O(N/\epsilon+\log(1/\epsilon)/\epsilon)$, which means that $\vblock$ can be produced from a shared maximally entangled state with at most $O(N/\epsilon+\log(1/\epsilon)/\epsilon)$ bits of communication.

	Recalling the definition of $\vblock$ in Equation \ref{eq:defphiprime}, let us define $\rho'_j \equiv Q_j \varphi  Q_j$, so that it is clear that $\vblock = \sum_j \rho'_j$.  It is also clear that $\rho'_j$ is not only PSD, but also an un-normalized pure state, because

	\[\rho_j' \equiv  Q_j \varphi  Q_j  =  Q_j \ket{\varphi} \bra{\varphi} Q_j. \]

	From the definition of $Q_j$ in Equation \ref{eq:defpij} we have that:

	\[ Q_j \ket{\varphi}  =  \sum_{B_s < l  \leq B_b }  \ket{\varphi_l}, \]

	Where the index limits are 
	
	\begin{align*}
	&B_s \equiv 2 (2((j-1) \cdot \lceil 1/\eps \rceil + k') - 1) \cdot \B\\
	&B_b \equiv 2 (2(j \cdot \lceil 1/\eps \rceil + k') - 1) \cdot \B.
	\end{align*}

	We know from the definition in Equation \ref{eq:defphij} that the $\ket{\varphi_l}$ are orthogonal to each other, and that each $\ket{\varphi_l}$ has Schmidt coefficients bounded by  $ 2^{-l N +1} \geq \lambda_i > 2^{- l N - 1} $.  Thus, it is immediate that $\rho_j'$ has spread at most $(B_b - B_s)N+4 = 2\lceil 1/\eps \rceil \B N +4 = O(N/\epsilon+\log(1/\epsilon)/\epsilon) $, where the last equality follows because $\B = 30 + 2 \left \lceil \frac{\log(1/\epsilon)}{N} \right \rceil$. Therefore $\vblock$ is a normalized mixture of states with spread at most $O(N/\epsilon+\log(1/\epsilon)/\epsilon)$.
	
	Consider the normalized version of $\rho_j'$, which is still a pure state of spread at most $O(N/\epsilon+\log(1/\epsilon)/\epsilon)$ it is clear that this state has Earthmover distance at most $O(N/\epsilon+\log(1/\epsilon)/\epsilon)$ from the nearest maximally entangled state  (simply move all of the weight onto Schmidt coefficients of the size of the smallest Schmidt coefficient, which can be done by moving all the weight a distance less than or equal to the spread).  It follows, by using Theorem \ref{thm::earthmover} that there is a protocol which prepares the normalized version of $\rho_i'$ from EPR pairs, with only $O(N/\epsilon+\log(1/\epsilon)/\epsilon)$ bits of communication (we note that this line of the proof could also have been established using result from  \cite{HL02,HW02}, for example).  Now the state $\vblock \equiv \sum_i \rho_i'$ can be prepared by applying this same protocol in superposition over $i$ (with the probability $\tr(\rho_i')$ assigned to each $i$), and then tracing out over the $i$ register.  Thus $\vblock$ can be prepared starting from EPR pairs with $O(N/\epsilon+\log(1/\epsilon)/\epsilon)$ bits of communication.

\end{proof}

\section{Proof of Lemma \ref{lem::quadpartite}} \label{pf::lem::quadpartite}

\begin{proof}
	
	Given two states $\ket\chi = \sum_{i\in X} \sqrt{\chi_i} \ket i \ot \ket i$ and $\ket\upsilon = \sum_{j\in Y} \sqrt{\upsilon_j} \ket j \ot \ket j$, and an arbitrary $\eps > 0$, let $\omega(i,j):  X \times Y \to \mathbb{R}_{\geq 0}$ be the joint distribution on $X \times Y$ which satisfies the $\ell_{\infty}$ Earth Mover conditions for $\ket \chi$ and $\ket \ups$, and acheives the optimal earth mover bound $\dearth{\ket \chi}{\ket \upsilon}$.  That is, for all $i \in X$, $\sum_{j \in Y} \omega(i,j) = \chi_i$, for all $j \in Y$, $\sum_{i \in X} \omega(i,j) = \ups_j$, and $\omega(i,j) = 0$ whenever $| \log(\chi_i) - \log(\ups_j)| > \dearth{\ket \chi}{\ket \upsilon}$.
	
	Define 
	%	$$\ket \rho \equiv \sum_{j\in Y} \sum_{k  \in [2^{\lceil \dearth{\ket \chi}{\ket \upsilon} \rceil} +2]} \sqrt{\upsilon_j/2^{\lceil \dearth{\ket \chi}{\ket \upsilon} \rceil +2 }} \ket j \ot \ket k \ot \ket j \ot \ket k$$  
%	In other words, we define
 $\ket \rho \equiv \sum_{j\in Y} \sum_{k \in [2^{\lceil \dearth{\ket \chi}{\ket \upsilon} \rceil}+2]} \sqrt{\rho_{j,k}} \ket j \ot \ket k \ot \ket j \ot \ket k$, where 
	
	$$\rho_{j,k} \equiv \upsilon_j/2^{\lceil \dearth{\ket \chi}{\ket \upsilon} \rceil +2 }. $$

	We now define the intermediate state
$$ \ket \gamma \equiv \sum_{j\in Y} \sum_{k \in [2^{\lceil \dearth{\ket \chi}{\ket
      \upsilon} \rceil + 2}]} \sum_{r \in [2^{\lceil \dearth{\ket \chi}{\ket \upsilon}
    \rceil + 2}]} \sqrt{\gamma_{j,k,r}} \ket j \ot \ket k \ot \ket r \ot \ket j \ot \ket k
\ket r,$$
where  the Schmidt coefficients $\gamma_{j,k,r}$ are left unspecified for now.  
	
	In order to specify the Schmidt coefficients of the intermediate state $\ket \gamma$ as well as the Right Index-1 Flow from $\ket \chi$ to $\ket \gamma$, and the Left Index-1 Flow from $\ket \gamma$ to $\ket \rho$ we will first define ``bins" for the Schmidt coefficients of $\ket \ups$ as follows:
	
	For $l \in \mathbb{N} \cup \{0\}$ let $\Upsilon_l \equiv \{j \in Y: 2^{-l} \geq \ups_j \geq 2^{-(l+1)}  \}$, and $X_l \equiv  \{i \in X: 2^{-l} \geq \chi_i
	\geq 2^{-(l+1)}  \}$.  Define $\omega(X_m, \Upsilon_l) \equiv \sum_{(i,j) \in X_m \times \Upsilon_l} \omega(i,j)$. 
	
	\begin{fact}  \label{fact::bindist}
		If $|m -l| >  \dearth{\ket \chi}{\ket \upsilon} + 1$, then $\omega(X_m, \Upsilon_l) = 0$
	\end{fact} 
	
	\begin{proof}
		Given $i \in X_m$, and $j \in \Upsilon_l$ we have by definition that $2^{-l} \geq \ups_j \geq 2^{-(l+1)}$, and $2^{-m} \geq \chi_i \geq 2^{-(m+1)}$, and therefore that $|\log(\chi_i) - \log(\ups_j)| \geq |m-l|-1 > \dearth{\ket \chi}{\ket \upsilon}$, where the last equality follows by assumption.  It follows by definition of $\dearth{\ket \chi}{\ket \upsilon}$ and of $\omega$, that $\omega(i,j) = 0$.  Since this is true for all $(i,j) \in X_m \times \Upsilon_l$, the claim follows. 
	\end{proof}
	
	We will now specify an iterative, ``greedy" procedure to define the Schmidt coefficients $\gamma_{j,k,c}$ as a function of the  $\ket{\chi}$ and $\ket{\rho}$.

	For each $(m,l) \in \mathbb{N} \cup \{0\} \times \mathbb{N} \cup \{0\}$ such that $\omega(X_m, \Upsilon_l) > 0$ we first note that by Fact \ref{fact::bindist} that  $|m -l| <  \dearth{\ket \chi}{\ket \upsilon} + 1$.  Thus, for each $(i,j) \in X_m \times \Upsilon_l$, 
	
	$$\chi_i \geq 2^{-(m+1)} \geq 2^{-l -\dearth{\ket \chi}{\ket \upsilon} - 2 } \geq 2^{-l}/2^{ \lceil \dearth{\ket \chi}{\ket \upsilon} \rceil + 2 }  \geq  \ups_j/2^{\lceil \dearth{\ket \chi}{\ket \upsilon} \rceil +2 } \equiv \rho_{j,k}$$ for all $k \in[2^{\lceil \dearth{\ket \chi}{\ket \upsilon} \rceil +2 }] $. 
	
	\begin{algorithm} \label{alg:defgamma}
		\caption{1}
		\begin{algorithmic}[1]

			\State For all $i$ set $\text{temp}_i = \chi_i$
			
			\State Set $i_{m} = \min \{X_m\}$ for all $m$  
			
			\For{$l \in \mathbb{N} \cup \{0\}$}
			\State Set $j := \min\{Y_l\}$;
			\State Set $k = 0$;
			
			\State Set $\text{overflow} = 0$
			\For{$m \in \mathbb{N} \cup \{0\}$} 
			\If{$\omega(X_m, \Upsilon_l) > 0$} 
			
			\State Set $\text{temp}_{\omega} = \omega(X_m, \Upsilon_l)$
			
			\While{$\text{temp}_{\omega} > 0$ }			
			
			\If{$\sum_{r \leq \text{overflow}} \gamma_{j,k, r} < \rho_{j,k}$} 
			\While{$\text{temp}_{\omega} \geq \rho_{j,k} -  \sum_{r \leq \text{overflow}} \gamma_{j,k, r}$}
			
			\If{$k = 2^{\lceil \dearth{\ket \chi}{\ket \upsilon} \rceil + 2} $}
			\State Set $j = j+1$
			\State Set $\text{overflow} = 0$
			\State Set $k = 0$
			
			\EndIf

			\If{$\text{temp}_{i_m} < \rho_{j,k} -  \sum_{r \leq \text{overflow}} \gamma_{j,k, r}$}
			\State Set $\gamma_{j,k, \text{overflow}+1} = \text{temp}_{i_m}$
			\State Set $\text{temp}_{\omega} = \text{temp}_{\omega} - \text{temp}_{i_m}$
			\State Set $\text{temp}_{i_m} = 0$
			\State Add an edge in the flow graph from $i_m$ to $(j,k, \text{overflow}+1)$  
			\State Set $i_m = i_m +1$

			\State Set $\text{overflow} = \text{overflow}+1$ 
			\EndIf
			\If{$\text{temp}_{i_m} \geq \rho_{j,k} -  \sum_{r \leq \text{overflow}} \gamma_{j,k, r}$ and $\text{temp}_{\omega} \geq \rho_{j,k} -  \sum_{r \leq \text{overflow}} \gamma_{j,k, r}$}
			\State Set $\gamma_{j,k, \text{overflow}+1} = \rho_{j,k} - \sum_{r \leq \text{overflow}} \gamma_{j,k, r}$
			\State Set $\text{temp}_{\omega} = \text{temp}_{\omega} -\gamma_{j,k, \text{overflow}+1} $
			\State Set $\text{temp}_{i_m} = \text{temp}_{i_m} - \gamma_{j,k, \text{overflow}+1}$
			
			\State Add an edge in the flow graph $G_{X,Z}$ from $i_m$ to $(j,k, \text{overflow}+1)$
			\State Set $k = k+1$ 
			\State Set $\text{overflow} = 0$
			\EndIf
			
			\EndWhile

			\algstore{myalg}
			
		\end{algorithmic}
	\end{algorithm}

	\begin{algorithm}
		% \ContinuedFloat
		\caption{ 1 (continued)}
		\begin{algorithmic}
			\algrestore{myalg}

			\If{$k = 2^{\lceil \dearth{\ket \chi}{\ket \upsilon} \rceil + 2} $} 
			\State Set $j = j+1$
			\State Set $\text{overflow} = 0$
			\State Set $k = 0$
			\EndIf

			\If{$\text{temp}_{\omega} < \rho_{j,k} -  \sum_{r \leq \text{overflow}} \gamma_{j,k, r}$}

			\If{$\text{temp}_{i_m} \leq \text{temp}_{\omega}$}
			\State Set $\gamma_{j,k, \text{overflow}+1} = \text{temp}_{i_m}$
			\State Set $\text{temp}_{\omega} = \text{temp}_{\omega} - \text{temp}_{i_m}$
			\State Set $\text{temp}_{i_m} = 0$
			\State Add an edge in the flow graph $G_{X,Z}$ from $i_m$ to $(j,k, \text{overflow}+1)$
			\State Set $i_m = i_m +1$

			\State Set $\text{overflow} = \text{overflow}+1$
			\EndIf
			\If{$\text{temp}_{i_m} \geq \text{temp}_{\omega}$}
			
			\State Set $\gamma_{j,k, \text{overflow}+1} = \text{temp}_{\omega}$
			\State Set $\text{temp}_{\omega} = 0$
			\State Set $\text{temp}_{i_m} = \text{temp}_{i_m} - \text{temp}_{\omega}$
			\State Add an edge in the flow graph $G_{X,Z}$ from $i_m$ to $(j,k, \text{overflow}+1)$

			\State Set $\text{overflow} = \text{overflow}+1$

			\EndIf
			
			\If{$k = 2^{\lceil \dearth{\ket \chi}{\ket \upsilon} \rceil + 2} $} 
			\State Set $j = j+1$
			\State Set $\text{overflow} = 0$
			\State Set $k = 0$
			
			\EndIf
			\EndIf
			
			\EndIf

			\EndWhile
			\EndIf
			\EndFor
			\EndFor

			%       \EndProcedure
		\end{algorithmic}
	\end{algorithm}

	One may check that Algorithm 1 defines Schmidt coefficients $\gamma_{j,k,r}$, satisfying 
	
	\[\sum_{j\in Y} \sum_{k \in [2^{\lceil \dearth{\ket \chi}{\ket \upsilon} \rceil + 2}]} \sum_{r \in [2^{\lceil \dearth{\ket \chi}{\ket \upsilon} \rceil + 2}] }  \gamma_{j,k,r}=  \sum_{i\in X} \chi_i = 1, \]
	
	as well as a Right Index-1 Flow from $\ket{\chi}$ to $\ket{\gamma}$, with degree at most $2^{\lceil \dearth{\ket \chi}{\ket \upsilon} \rceil + 2} \cdot 2^{\lceil \dearth{\ket \chi}{\ket \upsilon} \rceil + 2} = 2^{2 \lceil \dearth{\ket \chi}{\ket \upsilon} \rceil + 4}$.  In particular the Right Index-1 Flow from $\ket{\chi}$ to $\ket{\gamma}$ is constructed in Algorithm 1 by iteratively adding edges to form the bipartite flow-graph $G_{X, Z}$ where $Z \equiv (Y, [2^{\lceil \dearth{\ket \chi}{\ket \upsilon} \rceil + 2}], [2^{\lceil \dearth{\ket \chi}{\ket \upsilon} \rceil + 2}])$.  Each line in the pseudocode which reads ``Add an edge in the flow graph from $i_m$ to $(j,k, \text{overflow}+1)$", or similar, adds a single edge to the graph $G_{X,Z}$ and the union of all these edges forms the bipartite flow $G_{X,Z}$ between $X$ and $Z$.  Furthermore, for the $\gamma_{j,k,r}$ defined by Algorithm 1,

	\[ \sum_{r \in [2^{\lceil \dearth{\ket \chi}{\ket \upsilon} \rceil + 2}] }  \gamma_{j,k,r}=  \rho_{j,k}, \]
	
	so that there is a Left Index-1 flow from $\ket \gamma$ to $\ket \rho$ defined by a bipartite graph between the Schmidt coefficients of $\ket \gamma$ and $\ket \rho$ respectively, in which, for every $(j,k,r) \in Y \times [2^{\lceil \dearth{\ket \chi}{\ket \upsilon} \rceil + 2}] \times [2^{\lceil \dearth{\ket \chi}{\ket \upsilon} \rceil + 2}]$, there is an edge from $\gamma_{j,k,r}$ to $\rho_{j,k}$ of weight $\gamma_{j,k,r}$.  This Left Index-1 flow then clearly has degree $2^{\lceil \dearth{\ket \chi}{\ket \upsilon} \rceil + 2}$.

	Finally, recall that, 
	
	$$\sum_{k \in [2^{\lceil \dearth{\ket \chi}{\ket \upsilon} \rceil + 2}]}\rho_{j,k} = \sum_{k \in [2^{\lceil \dearth{\ket \chi}{\ket \upsilon} \rceil + 2}]} \upsilon_j/2^{\lceil \dearth{\ket \chi}{\ket \upsilon} \rceil +2 } = \nu_j$$
	
	So, by very similar reasoning, there is a Left Index-1 flow from $\ket{\rho}$ to $\ket{\nu}$ with degree exactly $2^{\lceil \dearth{\ket \chi}{\ket \upsilon} \rceil + 2}$.

\end{proof}

\pagebreak

\section*{Acknowledgments}
AWH was funded by NSF grants CCF-1452616, CCF-1729369, PHY-1818914 and ARO contract
W911NF-17-1-0433.  MC was supported at MIT by an Akamai Fellowship, and at the IQC by Canada's NSERC and the Canadian 
Institute for Advanced Research (CIFAR), and through funding provided to 
IQC by the Government of Canada and the Province of Ontario.

\bibliographystyle{hyperabbrv}
\bibliography{refs}

\begin{thebibliography}{10}

\bibitem{AHLNSZ}
D.~Aharonov, A.~W. Harrow, Z.~Landau, D.~Nagaj, M.~Szegedy, and U.~Vazirani.
\newblock Local tests of global entanglement and a counterexample to the
  generalized area law.
\newblock In {\em Foundations of Computer Science (FOCS), 2014 IEEE 55th Annual
  Symposium on}, pages 246--255, Oct 2014,
  \href{http://arxiv.org/abs/1410.0951}{{\ttfamily arXiv:1410.0951}}.

\bibitem{BBPS96}
C.~H. Bennett, H.~J. Bernstein, S.~Popescu, and B.~Schumacher.
\newblock Concentrating partial entanglement by local operations.
\newblock {\em Phys. Rev. A}, 53:2046--2052, 1996,
  \href{http://arxiv.org/abs/quant-ph/9511030}{{\ttfamily
  arXiv:quant-ph/9511030}}.

\bibitem{BDHSW-qrst}
C.~H. Bennett, I.~Devetak, A.~W. Harrow, P.~W. Shor, and A.~Winter.
\newblock The quantum reverse {S}hannon theorem and resource tradeoffs for
  simulating quantum channels.
\newblock {\em IEEE Trans. Inf. Theory}, 60(5):2926--2959, May 2014,
  \href{http://arxiv.org/abs/0912.5537}{{\ttfamily arXiv:0912.5537}}.

\bibitem{DH05}
S.~Daftuar and P.~Hayden.
\newblock Quantum state transformations and the schubert calculus.
\newblock {\em Annals of Physics}, 315:80--122, 2005,
  \href{http://arxiv.org/abs/quant-ph/0410052}{{\ttfamily
  arXiv:quant-ph/0410052}}.

\bibitem{Har-spread}
A.~W. Harrow.
\newblock Entanglement spread and clean resource inequalities.
\newblock In P.~Exner, editor, {\em {XVIth Int. Cong. on Math. Phys.}}, pages
  536--540. World Scientific, 2009,
  \href{http://arxiv.org/abs/0909.1557}{{\ttfamily arXiv:0909.1557}}.

\bibitem{HL07}
A.~W. Harrow and D.~W. Leung.
\newblock A communication-efficient nonlocal measurement with application to
  communication complexity and bipartite gate capacities.
\newblock {\em IEEE Trans. Inf. Theory}, 57(8):5504--5508, 2011,
  \href{http://arxiv.org/abs/0803.3066}{{\ttfamily arXiv:0803.3066}}.

\bibitem{HL02}
A.~W. Harrow and H.-K. Lo.
\newblock A tight lower bound on the classical communication cost of
  entanglement dilution.
\newblock {\em IEEE Trans. Inf. Theory}, 50(2):319--327, 2004,
  \href{http://arxiv.org/abs/quant-ph/0204096}{{\ttfamily
  arXiv:quant-ph/0204096}}.

\bibitem{HW02}
P.~Hayden and A.~Winter.
\newblock On the communication cost of entanglement transformations.
\newblock {\em Phys. Rev. A}, 67:012306, 2003,
  \href{http://arxiv.org/abs/quant-ph/0204092}{{\ttfamily
  arXiv:quant-ph/0204092}}.

\bibitem{JRS08}
R.~Jain, J.~Radhakrishnan, and P.~Sen.
\newblock Optimal direct sum and privacy trade-off results for quantum and
  classical communication complexity, 2008,
  \href{http://arxiv.org/abs/0807.1267}{{\ttfamily arXiv:0807.1267}}.

\bibitem{JungeP11}
M.~Junge and C.~Palazuelos.
\newblock Large violation of bell inequalities with low entanglement.
\newblock {\em Communications in Mathematical Physics}, 306(3):695--746, 2011,
  \href{http://arxiv.org/abs/1007.3043}{{\ttfamily arXiv:1007.3043}}.

\bibitem{LeungTW13}
D.~Leung, B.~Toner, and J.~Watrous.
\newblock Coherent state exchange in multi-prover quantum interactive proof
  systems.
\newblock {\em Chicago Journal of Theoretical Computer Science}, 11:1--18,
  2013,  \href{http://arxiv.org/abs/0804.4118}{{\ttfamily arXiv:0804.4118}}.

\bibitem{LP99}
H.-K. Lo and S.~Popescu.
\newblock The classical communication cost of entanglement manipulation: Is
  entanglement an inter-convertible resource?
\newblock 83:1459--1462, 1999,
  \href{http://arxiv.org/abs/quant-ph/9902045}{{\ttfamily
  arXiv:quant-ph/9902045}}.

\bibitem{Nielsen99a}
M.~A. Nielsen.
\newblock Conditions for a class of entanglement transformations.
\newblock 83:436--439, 1999,
  \href{http://arxiv.org/abs/quant-ph/9811053}{{\ttfamily
  arXiv:quant-ph/9811053}}.

\bibitem{Regev12}
O.~Regev.
\newblock Bell violations through independent bases games.
\newblock {\em Quantum Info. Comput.}, 12(1-2):9--20, Jan. 2012,
  \href{http://arxiv.org/abs/1101.0576}{{\ttfamily arXiv:1101.0576}}.

\bibitem{VW11}
T.~Vidick and S.~Wehner.
\newblock More nonlocality with less entanglement.
\newblock {\em Phys. Rev. A}, 83:052310, May 2011.

\end{thebibliography}

%\bibliography{../tex_headers/library}
\end{document}